%% file: Draft1703g.tex
\documentclass[a4paper,11pt]{article}
\pdfoutput=1 

\usepackage{jheppub} 
\usepackage{amsthm}
\usepackage{amsmath}
\usepackage{mathtools}
\usepackage{tabularx}
\usepackage{accents}
\usepackage{dsfont}
\usepackage{bm}
\usepackage{braket}
\newtheorem{theorem}{Theorem}

\newtheorem{definition}{Definition}
\usepackage{amsmath,amssymb}
\usepackage{amsthm}
\usepackage[T1]{fontenc}    
\usepackage{setspace}
\usepackage{tikz}

\usetikzlibrary{shapes.misc}
\usetikzlibrary{calc,angles,quotes,babel}
\usetikzlibrary{intersections}
\usetikzlibrary{hobby}
\usetikzlibrary{decorations.pathreplacing} 
\usetikzlibrary{positioning,shadings} 
\usepackage[automark]{scrpage2}
\usetikzlibrary{arrows.meta}
\tikzset{
  >={Triangle[length=0pt 3*8,width=0pt 8]},
}
\usepackage{varwidth}
\newcommand\Umbruch[2][4cm]{\begin{varwidth}{#1}\centering#2   \end{varwidth}}
\newcommand{\tr}{{\rm Tr}}

\input{./gmydefs.tex}

\title{Group Field Theory and Holographic Tensor Networks: Dynamical Corrections to the Ryu-Takayanagi formula}

\author[b]{Goffredo Chirco,}
\author[a]{Alex Goe{\ss}mann,}
\author[b,c]{Daniele Oriti}
\author[b]{and Mingyi Zhang}
\affiliation[a]{Institut f\"ur Mathematik, Technische Universit\"at Berlin, \\
Stra{\ss}e des 17.Juni 135, 10623 Berlin, Germany}
\affiliation[b]{Max-Planck-Institut f\"ur Gravitationsphysik (Albert-Einstein-Institut),\\Am M\"uhlenberg 1, 14476 Golm, Germany}
\affiliation[c]{Arnold-Sommerfeld-Center for Theoretical Physics, Ludwig-Maximilians-Universit\"at, \\ Theresienstrasse 37, 80333 M\"unchen, Germany.}

\emailAdd{goffredo.chirco(AT)aei.mpg.de}
\emailAdd{goessmann(AT)tu-berlin.de}
\emailAdd{daniele.oriti(AT)aei.mpg.de}
\emailAdd{mingyi.zhang(AT)aei.mpg.de}

\abstract{We introduce group field theory networks as a generalization of spin networks and of (symmetric) random tensor networks and provide a statistical computation of the R\'enyi entropy for a bipartite network state using the partition function of a simple interacting group field theory. The expectation value of the entanglement entropy is calculated by an expansion into stranded Feynman graphs and is shown to be captured by a Ryu-Takayanagi formula. 
For a simple interacting group field theory, we can prove the linear corrections, given by a polynomial perturbation of the Gaussian measure, to be negligible for a broad
class of networks.}

\begin{document} 
\maketitle
\flushbottom

\section{Introduction}

Tensor networks algorithms from condensed matter theory  \cite{11Oru, 12Bri, Cirac2009, Vidal2008, Verstraete2008} have recently experienced a massive impact in quantum gravity as new powerful tools for investigating the nature of spacetime at the Planck scale and its holographic properties. In the AdS/CFT framework, the Ryu-Takayanagi formula, together with the geometry/entanglement correspondence \cite{Ryu2006,VanRaamsdonk:2010pw,Faulkner2014,SwingleUniversality} have led to a new constructive approach to holographic duality, today further captured by the AdS/MERA conjecture \cite{SwingleEntanglement}, suggesting an interpretation of the geometry of the auxiliary tensor network decomposition of the quantum many-body boundary state as a representation of the dual spatial geometry. The use of tensor networks in this sense has produced a new constructive approach \cite{Hayden}, where the key entanglement features of some holographic theory can be captured by classes of tensor network states.  

In non-perturbative approaches to quantum gravity, including Loop Quantum Gravity (LQG) and Spin Foam Models \cite{Rovelli,Thiemann,01Ash,02Per} and their generalization in terms of Group Field Theories (GFT) \cite{03Ori,05Ori,06Gur}, pre-geometric quantum degrees of freedom are encoded in \emph{random combinatorial spin-network structures}, labeled by irreducible representation of $SU(2)$ and endowed with a gauge symmetry at each node. Such spin-network states can be understood as peculiar symmetric tensor networks \cite{Singh,Li2018}, and tensor network techniques have found a number of quantum gravity applications \cite{Dittrich2016,Delcamp2017,Dittrich2012,DSS2016}. A discrete spacetime and geometry is naturally associated with such structures, at a semi-classical level, and their quantum dynamics is related to (non-commutative) discrete gravity path integrals \cite{Baratin2010,Baratin2012,Han2013,HM2013}. The outstanding issue is then to show the emergence of {\it continuum} spacetime geometry and GR dynamics from the full quantum dynamics of the same pre-geometric degrees of freedom, which in fact describe a quantum spacetime as a peculiar sort of quantum many-body system \cite{07Ori,09Cao,Oriti2007}. In this sense, tensor network techniques have been largely exploited in relation to the problem of spin foam renormalization in the context of Loop Quantum Gravity \cite{Dittrich2016,Delcamp2017,Dittrich2012,DSS2016}, as well as quantitative tools to analyse the entanglement structure of spin-networks and look for classes of spin-network states with correlation and entanglement properties compatible with well behaved geometries in the semiclassical interpretation. 

More recently, tensor network representation schemes have been exploited to extract information on the non-local entanglement structure of the spin-network states and to understand the effects of the local gauge structure on the universal scaling properties of the holographic entanglement, in the background independent context \cite{Han}. Along this line, a precise dictionary between \emph{random} tensor networks and group field theory (GFT) states was defined by some of the authors in \cite{13Chi}, and used as a basis for a first derivation of the Ryu-Takayanagi formula \cite{Ryu2006} in a non-perturbative quantum gravity context. This dictionary also implied, under different restrictions on the GFT states, a correspondence between LQG spin-network states and tensor networks, and a correspondence between random tensors models \cite{RTM} and tensor networks. 

The generalized tensor field-theoretic formalism of GFTs has been successfully used to recover the statistical behavior of a special class of \emph{random} tensor network (RTN) states, in the large dimensional regime \cite{Hayden}. Idealized versions of RTNs, so-called pluri-perfect tensors, recently attracted a great deal of attention in the holographic context, as they can simultaneously satisfy the Ryu-Takayanagi (RT) formula for a subset of boundary states \cite{Pastawski}, they can be used to define bidirectional holographic codes \cite{YangHayden}, error correction properties of bulk local operators \cite{Almheiri2015}, and to investigate sub-AdS locality. In this sense, a fully developed dictionary hold strong promise for a new interplay between non-perturbative quantum gravity, AdS/CFT and quantum information theory. 

Conversely, the provided correspondence between group field theory (GFT) many-body states and large dimensional random tensor network, has allowed to reproduce standard techniques of random state averaging in the non-perturbative quantum gravity context, by means of a mapping to the evaluation of the partition function of a simple group field theory, understood as a peculiar statistical quantum many-body system. In the preliminary study provided in \cite{13Chi}, this analysis was limited to the case of a non-interacting GFT model. For a given GFT model, interaction kernels give the vertices of the spin foam, which can be naively seen as nodes of the bulk tensor network corresponding with the Feynman diagrams of the theory. In this sense, it is natural to wonder whether the presence of ``bulk'' interaction may leave an imprint on the expression of the holographic area law, in relation to deformation of the minimal surface.

In this paper, some of the results presented in \cite{13Chi} are re-derived and presented from a different perspective, hoping that this will facilitate their comprehension from a deeper and different angle. We then perform the calculation of the entanglement entropy for the simple interacting group field theory, corresponding to Boulatov's model for topological $3d$ gravity \cite{Boulatov, Ooguri}. From the statistical-mechanics point of view, the interacting model realizes a general non-Gaussian probability distribution over random tensor networks, which can be further exploited to characterize deviations from the perfect tensor behavior. Section 2 introduces the statistical treatment of the group field theory field, i.e. the basic dynamical field of a given GFT model, focussing on the relation, in terms of entanglement, between tensor fields and GFT fields. In Section 3 the notion of entanglement entropy for a bipartite network state is quantified by the measure of R\`enyi entropy and it is shown how, via replica trick, this quantity can be computed in terms of two auxiliary variables, $Z_0^{(N)}$ and $Z_A^{(N)}$, respectively defining the $N$th-power of the whole and reduced tensor network density matrix partition functions, graphically described in terms of stranded contraction patterns.  

Building on the group field-theoretic description, we then consider a statistical evaluation of the entanglement entropy, which, in the assumption of large dimension of the tensorial leg space $\mathcal{H}$, reduces by the concentration of measure phenomenon to the computation of a ratio of partition functions expectation values. The quantities $Z_0^{(N)}$ and $Z_A^{(N)}$ can be expanded in Feynman amplitudes corresponding to stranded diagrams. Their asymptotic behavior for large $\text{dim}[\mathcal{H}]$ is, however, captured only by the diagrams with maximal divergence degree. In Section 4 we define strategies to identify such maximal divergent diagrams, where we derive bounds of the divergence degree by the topology of the Feynman graphs. We study the maximum face number for different types of networks in a class of diagrams that we call \emph{locally averaged diagrams}. The maximally divergent diagrams in this class are analysed by first restricting to the free sector of the theory. In this setting, for the case of a unique minimal surface separating two boundary regions the mamxima are unique. Maximal divergent patterns for the case of multiple minimal surfaces are further discussed in Appendix A. We therefore consider the effect of a specific class of interactions in Section 5, by looking at a minimal perturbation of the free case given by a single interaction process. In this setting, we determine the maximal face numbers for the case of interaction kernels involving nodes incident to minimal surfaces. We find that the divergence degree of the diagrams induced by single local interaction processes is lower than in the free theory case, and we can find maximal divergent pattern only for a set of situations, where the network graph can be coarse-grained to a nontrivial tree structure.

\section{Random Group Fields and Tensor Networks}
\label{SECRandomTN}
Group field theories (GFTs) are quantum field theories defined on product spaces of groups, defined by combinatorially non-local kernels \cite{OriApproach, 03Ori, 05Ori}. GFTs provide a higher-rank generalization of matrix models with quantum states given by regular d-valent, graphs\footnote{including disconnected ones} labeled by group or Lie algebra elements, which can be equivalently represented as $(d-1)$-cellular complexes. The \emph{quantum dynamics} is defined by a vacuum partition function, whose perturbative expansion gives a sum of Feynman diagrams dual to $d$-cellular complexes of arbitrary topology. The Feynman amplitudes for these discrete histories can be written either as spin foam models or as simplicial gravity path integrals \cite{02Per,Ambjorn2005}. 
In the forthcoming derivation of the RT formula, we will consider GFTs defined in terms of a
(complex) bosonic field $\phi(g_1,\dots, g_d)$ on $G^{\times d}/G$, specifying to the case  $d=3$ and $G=SU(2)$. 

As a first step, we shall introduce group fields as generalizations of tensors, with the aim of strengthening the concept of entanglement as a unifying construction principle for both. We then introduce a generalization of random tensors in terms of GFT fields and derive group field network states as random variables dependent on field configurations.

\subsection{From Tensors to Group Fields}
A rank-$d$ tensor is an array of $N^d$ complex numbers\footnote{These numbers can be understood as coefficients in a basis decomposition, where for each rank a basis of cardinality $N$ is chosen.}, which is modeled by a field on the $d$th cartesian product of an index set with cardinality $N$. Each such index set can be enriched to represent the cyclic group $\mathbb{Z}_N=\{\ket{1},...,\ket{N}\}$, with the group relation induced by:
\begin{align}
\ket{k}\circ\ket{l}:=\ket{(k+l)\mod\,N}\quad \forall \ket{k},\ket{l}\in \mathbb{Z}_N
\end{align}  
Following this intuition, we define a rank-$d$ tensor as a field on a product group and generalize it afterwards to the case of more general group fields. 
\begin{definition}
	A rank-$d$ tensor $\ket{T}$ with index cardinality $N$ is a complex field on $d$ copies of the cyclic group $\mathbb{Z}_N$:
	\label{DEFtensor}
	\begin{align}
	\ket{T}:\mathbb{Z}_N^{\times d}\rightarrow \mathbb{C} \notag
	\end{align}
\end{definition}
Let $\mathcal{H}_{d,N}$ be the space of tensors with fixed rank $d$ and index cardinality $N$. Neglecting the structure of the cyclic group, $\mathcal{H}_{d,N}$ is reduced to $\mathbb{C}^{N^d}$. The linear structure, the scalar product and the completeness of $\mathbb{C}^{N^d}$ establish $\mathcal{H}_{d,N}$ to be a Hilbert space. A basis of $\mathcal{H}_{d,N}$ is chosen by $\ket{i_1,...,i_d}$, defined as:
\begin{align}
\label{DEFtensorbase}
\ket{i_1,...,i_d}(j_1\times...\times j_d)= \delta_{i_1,j_1}\cdot...\cdot\delta_{i_d,j_d}
\end{align}
With respect to this basis, we decompose a tensor $\ket{T}$ into its components $T_{i_1...i_d}$, which introduces an isomorphism to $\mathbb{C}^{N^d}$:
\begin{align}
\ket{T}=:\sum_{i_1,...,i_d \in \mathbb{Z}_N}T_{i_1...i_d}\ket{i_1,...,i_d}
\label{tensortoarray}
\end{align}
By decomposition of the basis elements in (\ref{DEFtensorbase}), we factorize the Hilbert space $\mathcal{H}_{d,N}$ in $d$ spaces $\mathcal{H}^{(l)}_{1,N}$, which we refer to as \textit{leg spaces} :
\begin{align}
\label{DEFtensordec}
\mathcal{H}_{d,N}=\bigotimes\limits_{l=1,..,d} \mathcal{H}^{(l)}_{1,N}\quad,\quad \ket{i_1,...,i_d}=\ket{i_1}\times...\times \ket{i_d}
\end{align}
Each leg space $\mathcal{H}^{(l)}_{1,N}$ has an induced Hilbert structure and a basis $\{\ket{i_l}\}$, whose product forms the basis elements (\ref{DEFtensorbase}) of the tensor space $\mathcal{H}_{1,N}$. 
A generic tensor $\ket{T}$ is given by an arbitrary superposition of the chosen basis elements and generally does not decompose as a product state in the leg spaces. In this case, we say that $\ket{T}$ is \textit{entangled} with respect to the given Hilbert space factorization. Identifying $\mathcal{H}_{d,N}$ with the space $\mathbb{C}^{N^d}$, this amounts to an impossible decomposition of an array $T_{i_1...i_d}$ into a product of one-dimensional arrays $\prod_l T^{(l)}_{i_l}$. 
To visualise this behavior, we model the state $\ket{T}$ in figure \ref{FIGtensor} as a node with associated $d$ ordered open legs $l$ representing the Hilbert spaces $\mathcal{H}_{1,N}^{(l)}$. 
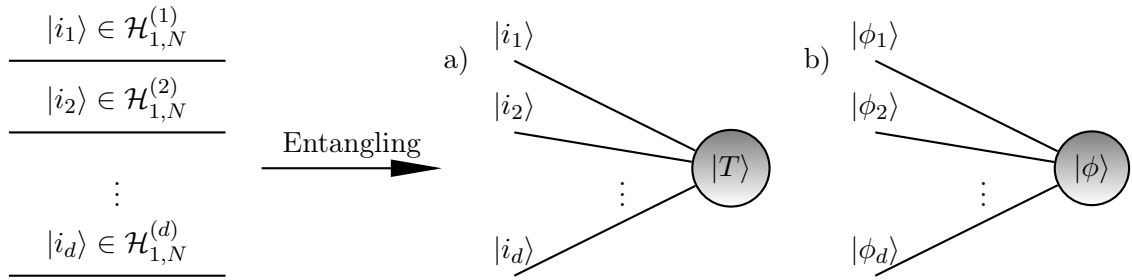
\begin{figure}[t]
	\centering
\begin{tikzpicture}[thick,scale=0.95]
\node (V) at (3.2,1.5) {a)};
\node (V) at (8.2,1.5) {b)};
\draw[-] (-3,1.5) -- (0,1.5) node[midway,above]{$\ket{i_1}\in \mathcal{H}_{1,N}^{(1)}$};
\draw[-] (-3,0.5) -- (0,0.5) node[midway,above]{$\ket{i_2}\in \mathcal{H}_{1,N}^{(2)}$} ;
\node (A) at (-1.5,-0.25)  {$\vdots$};
\draw[-] (-3,-1.5) -- (0,-1.5) node[midway,above]{$\ket{i_d}\in \mathcal{H}_{1,N}^{(d)}$} ;
\draw[->] (0.5,0) -- (3,0) node[midway,above]{Entangling} ;
\node (A) at (7,0) [circle,shade,draw] {$\ket{T}$};
\draw[-] (A) -- (4,1.5) node[above]{$\ket{i_1}$} ;
\draw[-] (A) -- (4,0.5) node[above]{$\ket{i_2}$} ;
\node (b) at (5.5,-0.25)  {$\vdots$};
\draw[-] (A) -- (4,-1.5) node[above]{$\ket{i_d}$};

\node (A) at (12,0) [circle,shade,draw] {$\ket{\phi}$};
\draw[-] (A) -- (9,1.5) node[above]{$\ket{\phi_1}$} ;
\draw[-] (A) -- (9,0.5) node[above]{$\ket{\phi_2}$} ;
\node (b) at (10.5,-0.25)  {$\vdots$};
\draw[-] (A) -- (9,-1.5) node[above]{$\ket{\phi_d}$};
\end{tikzpicture}
	\caption{Superposition of $d$ leg space product states $\ket{i_1}\times...\times \ket{i_d}$ with the weights by $T_{i_1...i_d}$ results in the tensor state $\ket{T}$. The nonentangled product states are visualized by nonconnected legs, in contary to the entangled node states a) $\ket{T}$ and more general group fields b) $\ket{\phi}$ as one connected component.}
	\label{FIGtensor}
\end{figure}

Moving one step further, one can consider the extension of the index set from a discrete cyclic group to a general locally compact group $G$. Although this results in dealing with infinite dimensional leg spaces $\mathcal{H}$, we can use the left Haar measure $\mu$ of $G$ \footnote{Since $G$ is locally compact, there exists a left-invariant Haar measure \cite{Werner}.} to lift the decompositions (\ref{tensortoarray}) to more general integrals.

\begin{definition}
	\label{DEFgroupfields}
	Let $G$ be a locally compact group with the normed left Haar measure $\mu$ of its elements. A \textbf{rank-$d$ group field} $\phi$ is a $\mu^{\times d}$-integrable field over $d$ copies of the group $G$:
	\begin{align}
	\phi:G^{\times d}\rightarrow \mathbb{C}\quad,\quad(g_1\times...\times g_d)\rightarrow \phi(g_1,...,g_d)\in \mathbb{C}\\
	\text{ such that } ||\phi||^2:=\int_{G^{\times d}} \phi(g_1,..,g_d)\overline{\phi(g_1,..,g_d)} d\mu^{\times d}(g_1,...,g_d) < \infty
	\end{align}
\end{definition}  

If we take $G$ as the cyclic group $\mathbb{Z}_N$, the notion of group fields reduces to tensors, where integrability with respect to the discrete Dirac measure $\mu$ is satisfied for all fields considered in definition \ref{DEFtensor}. Analogously to the space $\mathcal{H}_{d,N}$ of tensors, the space of integrable group fields carries a Hilbert structure induced by the space $L^2(G^{\times d},\mu^{\times d})$ of $l_2$-normed fields. The demand of integrability for $\phi$ ensures the finiteness of the inner product on $L^2(G^{\times d} ,\mu^{\times d})$ by the Cauchy-Schwarz inequality:
\begin{align}
\braket{\phi | \psi}:=\int_{G^{\times d}}\phi(g_1,...,g_d) \overline{\psi(g_1,...,g_d)} d\mu^{\times d}(g_1,...,g_d) \leq ||\phi||\cdot ||\psi|||<\infty
\end{align}

A group field $\phi$ is therefore understood as a state $\ket{\phi}$ in the infinite dimensional Hilbert space $L^2(G^{\times d},\mu^{\times d})$. Dirac's delta symbols $\bra{g_1,...,g_d}$ in the associated dual space, which we used to define a basis decomposition (\ref{DEFtensorbase}) in case of tensors, do not correspond to elements in $L^2(G^{\times d},\mu^{\times d})$. For the following we understand them as distributions acting on the group field space $L^2(G^{\times d},\mu^{\times d})$:
\begin{align}
\label{DEFgroupbase}
\bra{g_1,...,g_d}\in \left((L^2(G^{\times d},\mu^{\times d})\right)^{*}\quad ,\quad
\braket{g_1,...,g_d|\phi}=\phi(g_1,...,g_d)
\end{align} 
Also the Hilbert space $L^2(G^{\times d},\mu^{\times d})$ admits by construction a factorization into leg Hilbert spaces $L^2(G,\mu)^{(l)}$, which correspond to a product decomposition of the distributions (\ref{DEFgroupbase}) into distributions $\bra{g_l}$ affecting single variables:
\begin{align}
L^2(G^{\times d},\mu^{\times d})=\bigotimes_{l=1,...,d}L^2(G,\mu)^{(l)}\quad,\quad \bra{g_1,...,g_d}=\bra{g_1}\times...\times \bra{g_d}
\end{align}
With the established factorization of the group field Hilbert space, we are now able to generalize the notion of entanglement with the use of Dirac distributions instead of basis decompositions:

\begin{definition}
A group field theory state $\ket{\phi}\in L^2(G^{\times d},\mu^{\times d})$ is called \textbf{unentangled} with respect to the factorization $L^2(G^{\times d},\mu^{\times d})=\bigotimes_{l=1,...,d}L^2(G,\mu)^{(l)}$, if there is a collection of group fields $\{\ket{\phi_{l}}\in L^2(G,\mu)\}_{l=1}^d$ such that $\ket{\phi}=\ket{\phi_1}\times...\times\ket{\phi_d}$, i.e. it holds:
\begin{align}
\braket{g_1,...,g_d|\phi}=\prod_{l=1}^d \braket{g_l|\phi_{l}}\quad \forall g_1\times...\times g_d\in G^{\times d}
\label{DEFunent}
\end{align}
If there is no such collection, the field $\ket{\phi}$ is called \textbf{entangled}.
\end{definition}

Entanglement with respect to the decomposition of group field spaces into a collection of field spaces on smaller product groups is thus an analogous concept to the case of tensors on discrete groups, with the difference just lying in a more general perspective by distributions compared to basis decompositions. For the sake of simplicity we will treat the Dirac distributions $\bra{g_1,...,g_d}$ in the following as elements $\ket{g_1,...,g_d}$ of the space  $L^2(G^{\times d},\mu^{\times d})$, thus their evaluations (\ref{DEFgroupbase}) are understood as a proper scalar product.

\subsection{Free Tensor Models}
\label{SECgftqft}
The central idea of a tensor field theory is to implement random distributions on field spaces, which amounts in the case of the group field space $L^2(G^{\times d},\mu^{\times d})$ to a probability measure with density $d\nu(\phi)$ \cite{06Gur} \footnote{The definition of $d\nu(\phi)$ is usually given by a field measure at each element of $X$ and gets therefore difficult if the cardinality of the set $X$ is not finite \cite{06Gur}. The reason for this lies in the structure of the space $F(X)$, the space of complex functions on $X$. If and only if $|X|<\infty$, the dimension of $F(X)$ is finite and we can construct measures using the Lebesgue measure. For infinite cardinality of $X$, this is no longer possible, and the use of other methods like limit processes of measures is required \cite{Zee}.}. Measurable observations of a quantum field theory correspond to random variables, so-called observables, $O(\phi)$ with a probability character induced by the field measure $\nu$.

Well controllable field theories are Gaussian probability measures, which correspond to \emph{free field theories}. For the sake of simplicity, let us introduce Gaussian probability measures on the space of tensors $\mathcal{H}_{d,N}$, since we can exploit basis decompostions in this case. Associated with each Gaussian probability measure is a \emph{covariance} $\mathcal{C}$, which is an endomorphism in the space of field configurations. In the case of the finite dimensional configuration space $\mathcal{H}_{d,N}$, $\mathcal{C}$ is described by a matrix $\mathcal{C}(i,j)$:
\begin{align}
\mathcal{C}\in L(\mathcal{H}_{d,N},\mathcal{H}_{d,N}) \quad ,\quad T_{i_1...i_d}\rightarrow \sum_{j_1,...,j_d} \mathcal{C}(i_1,...,i_d,j_1,...,j_d)T_{j_1...j_d}
\end{align}

\begin{definition}
	\label{DEFfreetensormodel}
	Let $\mu_{\mathcal{C}}(T)$ be a probability measure of fields $T$ on $X=\mathbb{Z}_N^{\times d}$. $\mu_{\mathcal{C}}(T)$ is called \textbf{free Tensor model} with covariance $\mathcal{C}$, if the only nonvanishing expectation of observables are given by linear combinations of so-called $2p$-point Green functions:
	\begin{align}
	\braket{\prod\limits_{k=1}^p T_{i^{(k)}_1...i^{(k)}_d}\overline{T}_{j^{(k)}_1...j^{(k)}_d}}=\sum\limits_{\pi \in \mathcal{S}_N} \prod\limits_{k=1}^p \mathcal{C}(i^{(k)}_1,...,i^{(k)}_d,j^{(\pi(k))}_1,...,j^{(\pi(k))}_d)
	\label{permutationdecompositon}
	\end{align}
\end{definition}
The $2p$-point Green functions have a direct physical interpretation in terms of particle propagation. The $p$ fields $T_{i^{(k)}_1...i^{(k)}_d}$ located at their arguments $(i^{(k)}_1,...,i^{(k)}_d)\in \mathbb{Z}_N^{\times d}$ represent the incoming particles, which propagate to the outgoing particles $\overline{T}_{j^{(k)}_1...j^{(k)}_d}$ located at $(j^{(k)}_1,...,j^{(k)}_d)$. Each permutation $\pi \in S_N$ encodes the propagation of the $k$th incoming particle to the $\pi(k)$th outgoing particle, where the sum over all propagation possibilities is taken. A typical choice for the covariance, determining the structure of particle propagation, is given by the identification of each field argument:
\begin{align}
\label{deltakernel}
\mathcal{C}(i_1,...,i_d,j_1,...,j_d):=\prod_{l=1}^{d}\delta(i_l,j_l)
\end{align} 
Let us now assume the existence of a to $\mathcal{C}$ inverse covariance, that is a covariance $\mathcal{K}$ such that:
\begin{align}
\label{inversekernel}
\sum_{h_1,...,h_d=1}^{N} \mathcal{C}(i_1,...,i_d,h_1,...,h_d)\mathcal{K}(h_1,...,h_d,j_1,...,j_d)=\prod_{l=1}^{d}\delta(i_l,j_l)
\end{align}
This is for instance the case for the choice (\ref{deltakernel}). Using the inverse covariance $\mathcal{K}$, one can express the Gaussian Tensor Model by the Lebesgue measure $dT_{i_1...i_d}$ of the field value at each element of $\mathbb{Z}_N^{\times d}$ \cite{06Gur}. 	
\begin{theorem}\label{THEwicktensor}
	Given a by $\mathcal{K}$ invertible covariance $\mathcal{C}$, the associated Gaussian Tensor Model is given by the probability density:
	\begin{align}
	d\mu_\mathcal{C}(T)=\det(\mathcal{C}) \left[ \prod\limits_{\{h_l\}}^N \frac{dT(\{h_l\})d\overline{T}(\{h_l\})}{2\pi}  \right] e^{-\sum_{\{i_l\},\{j_l\}}T(\{i_l\})\mathcal{K}(\{i_l\},\{j_l\})\overline{T}(\{j_l\})}\\ \nonumber \quad \text{(Wick)}
	\end{align}
\end{theorem}
The exponential of the weight, which transforms the Lebesgue measures to the tensor model, is called the tensor model action $S[T]$. In the free Gaussian theory, the associated action $S_0[T]$ is called free:
\begin{align}
S_0[T]:=\sum_{\{i_l\},\{j_l\}}T(\{i_l\})\mathcal{K}(\{i_l\},\{j_l\})\overline{T}(\{j_l\})
\end{align}

\subsection{Free and perturbed Group Field Theory}
Allowing for groups with infinite cardinality, beyond the restriction to the cyclic group considered above, results in a significantly richer theory. The covariance $\mathcal{C}$ is generalized by its action on the Dirac distributions $\bra{g_1,...,g_d}$, which replaced in (\ref{DEFgroupbase}) the finite basis decomposition:
\begin{align}
\mathcal{C}:\bra{g_1,...,g_d} \rightarrow \int_{G^{\times d}}\prod_{i=1}^d d\overline{g}_i \mathcal{C}(g_1,...,g_d,\overline{g}_1,...,\overline{g}_d)\bra{\overline{g}_1,...,\overline{g}_d} \\
\mathcal{C}( \bra{g_1,...,g_d} )[\phi]= \int_{G^{\times d}}\prod_{i=1}^d d\overline{g}_i \mathcal{C}(g_1,...,g_d,\overline{g}_1,...,\overline{g}_d)\phi(\overline{g}_1,...,\overline{g}_d)
\end{align}
A Gaussian Group Field Theory with covariance $C$ is defined in analogy to definition \ref{DEFfreetensormodel} by replacing the cyclic group arguments $\{i_l,j_l\in \mathbb{Z}_N\}$  by general group arguments $\{g_l,h_l\in G\}$ and the tensors $T$ by general group fields $\phi$. The covariance $\mathcal{C}$ is invertible by the covariance $\mathcal{K}$, if $\mathcal{C} \circ \mathcal{K}=\text{Id}_{L^2(G^{\times d},G^{\times d})}$. In this case, one can associate a free action to the covariance $\mathcal{C}$, which, in analogy to the tensor model case, is given by 
\begin{align}
\label{DEFfreeactiongroupfields}
S_0[\phi]=\int_{G^{\times 2d}}[\prod_{l=1}^d d\mu(g_l)d\mu(h_l)]\, \phi(\{g_l\})\mathcal{K}(\{g_l\}\{h_l\})\overline{\phi}(\{h_l\})
\end{align}

Assuming the existence of a Gaussian measure $d\mu_{\mathcal{C}}(\phi)$ for the group field configurations $\phi$ associated with the covariance $\mathcal{C}$, we define the Lebesgue measure $[D\phi]$ in analogy to the tensorial case by
\begin{align}
d\mu_\mathcal{C}(\phi)=:[D\phi]e^{-S_0[\phi]} 
\label{DEFfreemeasure}
\end{align}
Pure Gaussian probability measures describe free field propagation, formalized by $2p$-point functions in definition \ref{DEFfreetensormodel}. In order to include field interactions, a perturbation of the free Gaussian measure is needed, which can be implemented by a control parameter $\lambda$: 
\begin{align}
\label{DEFperturbation}
S_0[\phi]\rightarrow S_0[\phi]+\lambda S_{\text{int}[\phi]},
\end{align}
The perturbed exponential weight of the probability distribution can be expanded in a series of manipulated free Gaussian probabilities, which can be interpreted as different orders of interactions:
\begin{align}
\label{intexpension}
[D\phi]e^{-S_0[\phi]-\lambda S_{\text{int}}[\phi]} & =[D\phi]e^{-S_0[\phi]}\cdot e^{-\lambda S_{\text{int}}[\phi]}\\ \nonumber
& =d\mu_C(\phi)\left[1-\lambda S_{\text{int}}[\phi]+\frac{\lambda^2}{2} (S_{\text{int}}[\phi])^2+\mathcal{O}(\lambda^3)\right]
\end{align}
The result is a d-dimensional combinatorially non-local quantum field theory living on a product group manifold \cite{03Ori}. By decomposition of observables $O(\phi)$ into $2p$-point functions and use of an analog version of Wicks theorem (theorem \ref{THEwicktensor}) expectation values $\braket{O}$ are expanded into series of terms with interpretation in terms of Feynman graphs. Due to the defining combinatorial structure by the action terms $S_0$ and $S_{\text{int}}$, the Feynman diagrams of the theory are dual to cellular complexes, and the perturbative expansion of the quantum dynamics defines a sum over random lattices of arbitrary topology. A similar lattice interpretation can be given to the quantum states of the theory. For group field theory models, where appropriate group theoretic data are used and specific properties are imposed on the states and quantum amplitudes, the same lattice structures can be understood in terms of simplicial geometries \cite{03Ori, BarOr2012, OriApproach, 13Chi}.

\subsection{Closure constraint for group fields}
\label{SECsymmetricgroupfields}

Group field theories provide a generalizing structure for non-perturbative approaches to Quantum Gravity \cite{03Ori,BarOr2012}. Their interpretation as simplicial quantum geometries, in terms of spin-networks \cite{13Chi}, relies on the restriction to \emph{gauge invariant states}, satisfying the so-called closure constraint
, given in generalization to arbitrary combinatorial dimension $d$ as:
\begin{align}
\label{DEFsymmetry}
\phi_0(hg_1,...,hg_d)\stackrel{!}{=}\phi_0(g_1,...,g_d) \quad \forall h\in G
\end{align}
For convenience in the forthcoming calculations, let us consider the imposition of such constraint for random group fields by splitting a general group field $\phi \in L^2(G^{\times d},\mu^{\times d})$ into a field $\phi_0$, where  the symmetry is realized by group averaging (\ref{DEFsymmetry}), and a gauge field $\chi$:
\begin{align}
\phi_0(g_1,...,g_d): & =\int_G \phi(hg_1,...,hg_d) d\mu(h) \\
\chi(g_1,...,g_d): & = \phi(g_1,...,g_d)-\phi_0(g_1,...,g_d)
\end{align}
The invariance of the Haar measure $\mu$ under shifts of the transformation variable $h$ ensures the symmetry (\ref{DEFsymmetry}) for the field $\phi_0$:
\begin{align}
\phi_0(\hat{h}g_1,...,\hat{h}g_d)& =\int_G \phi(h\hat{h}g_1,...,h\hat{h}g_d) d\mu(h)=\int_G \phi(hg_1,...,hg_d) d\mu(h\hat{h}^{-1}) \label{symmetryprojection}
\\
& =\phi_0(g_1,...,g_d)
\label{symmetrychecked}
\end{align}
We can then impose an equivalence relation $\sim $ in $L^2(G^{\times d},\mu^{\times d})$ and choose the symmetric elements $\phi_0$ as a representative of it:
\begin{align}
\nonumber
\phi^{(1)}\sim \phi^{(2)} & \iff \int_G \phi^{(1)}(hg_1,...,hg_d) d\mu(h)=\int_G \phi^{(2)}(hg_1,...,hg_d) d\mu(h)=:\phi_0\\
& \iff \phi^{(1)},\phi^{(2)}\in [\phi_0]
\label{DEFequivalenceclassphi}
\end{align}
The field $\chi$ it thereby redundant, in the sense that it labels the elements in the equivalence class $[\phi_0]$:
\begin{align}
\phi \in [\phi_0] \iff \exists \chi\in L^2(G^{\times d},\mu^{\times d}): \quad \phi=\phi_0+\chi \, \land \, \int_G \chi(hg_1,...,hg_d) d\mu(h)=0
\end{align}
We can represent each equivalence class $[\phi_0]$ by a $(d-1)$-dimensional group field. Therefore, we define an equivalence relation in $G^{\times d}$ to identify the relevant arguments of the representing field:
\begin{align}
\label{DEFequivalenceclassG}
 (g_1,...,g_d)\sim (\hat{g}_1,...,\hat{g}_d) & \iff \exists h\in G:\, (g_1,...,g_d)=(h\hat{g}_1,...,h\hat{g}_d)\\
 (H_1,...,H_{d-1},e)& :=(g_d^{-1}g_1,...,g_{d}^{-1}g_d)\in [(g_1,...,g_d)]
\end{align}
Since $\phi_0$ is constant for all elements in a class $[(g_1,...,g_d)]$ of arguments, it is already characterized by its values for the representative $(H_1,...,H_{d-1},e)$, which form an embedding of $G^{\times (d-1)}$ in $G^{\times d}$. We can thus think of $\phi_0$ as an $(d-1)$-field by restricting it to this embedding. Based on this observation, one can define $d$-dimensional Group Field Theories depending just on symmetric part $\phi_0$ of the described group fields $\phi$. We impose this dependence with the use of propagation and interaction kernels respecting the symmetries \cite{13Chi}:
\begin{align}
\label{DEFsymmetryK}
\mathcal{K}(\{g^{(1)}_i\}\{g^{(2)}_i\}) & =\int\limits_{G^{\times 2}}dh d\hat{h}\,\mathcal{K}(\{hg^{(1)}_i\}\{\hat{h}g^{(2)}_i\})\\
\label{DEFsymmetryV}
\mathcal{V}(\{g^{(1)}_i\}...\{g^{(d+1)}_i\}) & = \int\limits_{G^{\times (d+1)}}[\prod_{j=1}^{d+1}dh_j]\, \mathcal{V}(\{h_1g^{(1}_i\}...\{h_{d+1}g^{(d+1)}_i\})
\end{align}
The action constructed by convolution of group fields weighted by a kernel $\mathcal{K}$ satisfying (\ref{DEFsymmetryK}), and analogously for kernel $\mathcal{V}$ satisfying (\ref{DEFsymmetryV}), does just depend on the symmetric part $\phi_0$, since for arbitrary fields $\chi$ it holds:
\begin{align}
\nonumber
S_0[\phi_0+\chi]& = \int dg^{(1)}_i dg^{(2)}_i dh d\hat{h}\, \mathcal{K}(\{hg^{(1)}_i\}\{\hat{h}g^{(2)}_i\}) \phi(g^{(1)}_i)\overline{\phi}(g^{(2)}_i)\\
\nonumber
& = \int d(hg^{(1)}_i) d(\hat{h}g^{(2)}_i)\, \mathcal{K}(\{hg^{(1)}_i\}\{\hat{h}g^{(2)}_i\}) \int_G dh \phi(hg^{(1)}_i) \int_G d\hat{h} \phi(\hat{h}g^{(2)}_i)\\
& = S_0[\phi_0]
\end{align}
If the propagation kernel $\mathcal{K}$ has the demanded symmetry, it is not invertible. As a direct consequence, there is no Gaussian probability measure on the space of in general $d$-dimensional group fields, such that $\mathcal{K}$ arises as the inverse covariance. The reason for this lies in the gauge freedom by the choice of the field $\chi$, which does not affect the action. In order to get rid of these freedom, we define the action $S$ just on the space of symmetric fields $\phi_0$ and strip the gauge freedom $\chi$ of. This corresponds to the reduction of the elements in an equivalence class (\ref{DEFequivalenceclassphi}) of group fields to one representative, given by the effective $(d-1)$-dimensional group field $\phi_0$. 
The action $S_0[\phi_0]$ thus defines an effective dimensional reduction.

\subsection{Group field networks as observables}
We shall now discuss the construction of generalized tensor networks, realized by entanglement of individual group field states. By construction, such network states shall now be understood as random variables induced by the probabilistic character of their building blocks and interpreted as observables of the associated group field theory.

Let us first recall the graphical visualisation of random fields $\phi\in L^2(G^{\times d}, \mu^{\times d})$ as nodes representing the entanglement of the state with respect to the decomposition of $L^2(G^{\times d}, \mu^{\times d})$ into leg spaces $L^2(G, \mu):=\mathcal{H}$ (see figure \ref{FIGtensor}). Consider now a set of $V$ nodes, where each node $v\in V$ is dressed with a random field $\phi^{(v)} \in L^2(G^{\times d}, \mu^{\times d}):=\mathcal{H}^{\otimes d}$. This corresponds to a state $\Psi$ in the product of the node spaces:
\begin{align}
\Psi=\bigotimes_{v\in V}\phi^{(v)}\in \mathcal{H}^{d\cdot|V|}=\bigotimes_{v\in V}\mathcal{H}_v^{d}
\label{unentangledstate}
\end{align}

Instead of looking at the decomposition of node Hilbert spaces $\mathcal{H}^{d}$ into leg Hilbert spaces, we lift our focus towards decomposition of many-particle Hilbert spaces $\mathcal{H}^{d\cdot|V|}$ into node Hilbert spaces. With respect to such decompositions, the state $\Psi$ is by construction unentangled.

A simple method to construct network states, which are entangled with respect to the node decomposition of the many-particle Hilbert space, is given by the projection of the unentangled state in \ref{unentangledstate} from $\mathcal{H}^{d\cdot|V|}$ to an entangled state in the product leg space. We can visualize such a procedure as a gluing of the open tensor legs, which represents entanglement. 
Consider for instance the unentangled product of two node states $\phi^{(1)}\otimes \phi^{(2)}\in \mathcal{H}^{\otimes d}\otimes \mathcal{H}^{\otimes d}$. A \emph{specific} gluing functional is chosen by an integrated delta distribution $M$ acting on the $a$th leg of the first and the $b$th leg of the second node:
\begin{align}
\phi^{(1)}\otimes \phi^{(2)}\rightarrow M\left( \phi^{(1)}\otimes \phi^{(2)} \right):=\int_{G^{\times 2}}\delta(g^a_1g^b_2)\phi^{(1)}(g_1^1,...,g_1^d)\phi^{(2)}(g_2^1,...,g_2^d)dg^a_1dg_2^b
\label{twogluing}
\end{align} 
The resulting state can no longer be decomposed into a product of states in the open legs of the first and the second node, and is thus entangled. In the established graphical representation scheme of figure \ref{FIGtensor}, this corresponds to a link among the $a$th and $b$th leg connected node states, as sketched in figure \ref{FIGentangled} a).

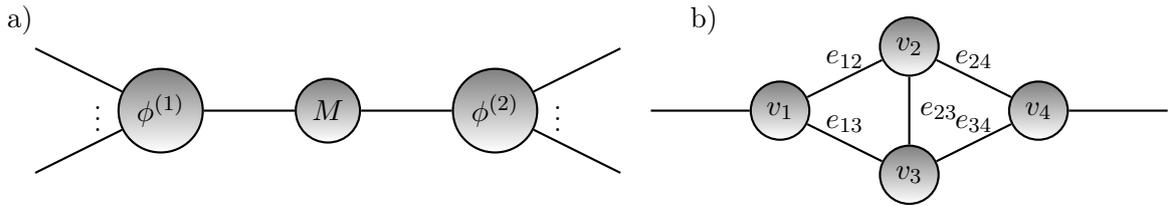
\begin{figure}
\begin{tikzpicture}[thick]
\node (AR) at (2,1.2)  {a)};
\node (AR) at (11,1.2)  {b)};
\begin{scope} [scale=0.55]
\node (N) at (11,0) [circle,shade,draw] {$M$};
\node (A) at (15,0) [circle,shade,draw] {$\phi^{(2)}$};
\draw[-] (A) -- (18,1.5);
\draw[-] (A) -- (N) ;
\node (b) at (16.5,0)  {$\vdots$};
\draw[-] (A) -- (18,-1.5);
\node (A) at (7,0) [circle,shade,draw] {$\phi^{(1)}$};
\draw[-] (A) -- (4,1.5);
\draw[-] (A) -- (N) ;
\node (b) at (5.5,0)  {$\vdots$};
\draw[-] (A) -- (4,-1.5);
\end{scope}
\begin{scope}[shift={(12,0)},scale=0.85]
\node (A1) at (0,0) [circle,shade,draw] {$v_1$};
\node (A2) at (2,1) [circle,shade,draw] {$v_2$};
\node (A3) at (2,-1) [circle,shade,draw] {$v_3$};
\node (A4) at (4,0) [circle,shade,draw] {$v_4$};
\draw[-] (A1) -- (A3) node[midway,above]{$e_{13}$};
\draw[-] (A1) -- (A2) node[midway,above]{$e_{12}$} ;
\draw[-] (A2) -- (A3) node[midway,right]{$e_{23}$} ;
\draw[-] (A4) -- (A3) node[midway,above]{$e_{34}$} ;
\draw[-] (A2) -- (A4) node[midway,above]{$e_{24}$} ;
\draw[-] (A1) -- (-2,0) ;
\draw[-] (A4) -- (6,0);
\end{scope}
\end{tikzpicture}
	\caption{a) Entangling two node states $\phi^{(1)}$ and $\phi^{(2)}$ by action of the functional $M$ on affected leg spaces results in an connected state. b) By repetition of the entangling procedure along all edges of a graph $\Gamma=(V,E\cup \partial \Gamma$) results in a network state $\Phi_{\Gamma}$.}
	\label{FIGentangled}
\end{figure}

It is useful to introduce a notion of link state $M$ \cite{13Chi}, to redefine the gluing transformation (\ref{twogluing}) as a scalar product in the leg spaces. In the basis given in \eqref{DEFgroupbase}, we have 
\begin{align}
\langle M| = \int_{G^{\times2}} dg_adg_b \, \delta(g_a\,g_b)\,\, \ket{g_a} \otimes  \ket{g_b} \quad \in \quad L^2(G^{\times 2}, \mu^{\times 2})
\end{align}
The gluing (\ref{twogluing}) then corresponds to a contraction of the link state $M$, such that
\begin{align}
\ket{\Phi_{12}}:= \braket{M|\phi^{(1)}\otimes \phi^{(2)}}
\end{align}
Notice that this is a stronger notion of gluing than the one used in GFT states, to define states associated with closed graphs, however corresponds to the standard gluing prescription for maximally entangled tensor networks states. More generally, one would consider a link convolution functional $M(g_a^{\dagger}\, h_l\, g_b)$ in the link Hilbert space, the product space of two leg spaces. With the adopted definition one effectively sets $h_l = e$ for all link $l\in\Gamma$. This assumption makes our state $\ket{\Psi_{\Gamma}} $ lying in the flat vacuum 
of loop quantum gravity \cite{Dittrich_2015}.	

\

With the established entangling projections, we are now ready to introduce general network states. Let there be an open graph $\Gamma=(V,E\cup \partial\Gamma)$ consistent of a vertex set $V$ and a set of $E$ edges incident with two vertices and a set $\partial\Gamma$ of open edges incident to single vertices. By iteratively projecting with $\{M^{(e)}\}_{e\in E}$ for each edge comprising two legs of adjacent vertex states $\phi^{(v)}$ transforms (\ref{unentangledstate}) into the network state:
\begin{align}
\ket{\Phi_{\Gamma}}=\bigotimes_{e\in E}\bra{M^{(e)}}\bigotimes_{v\in V}\ket{\phi^{(v)}} \quad\in \quad L^2(G^{\times |\partial \Gamma |}, \mu^{\times |\partial \Gamma |}) 
\end{align}
We thus entangled the group field vertices iteratively by edgewise projections, where the affected legs are determined by the graph $\Gamma$. The resulting network state $\ket{\Phi_{\Gamma}}$ is thus a collection of entangled group fields, where the connectivity of $\Gamma$ corresponds to the entanglement of its vertices (figure \ref{FIGentangled} b)).
%

\section{Holographic Entanglement Entropy for GFT Tensor Networks}

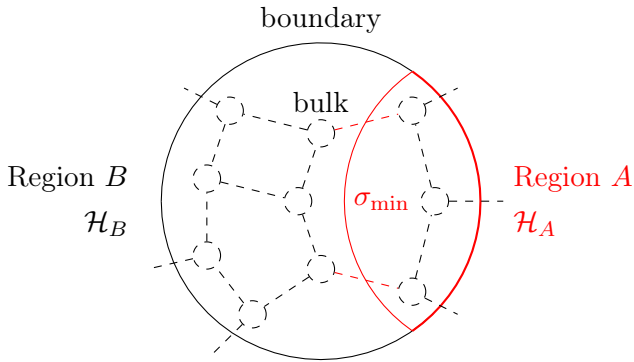
\begin{figure}[h]
\centering
\begin{tikzpicture}[scale=0.6]
\draw [domain=55:180] plot ({3.5*cos(\x)}, {3.5*sin(\x)});
\draw [domain=180:315] plot ({3.5*cos(\x)}, {3.5*sin(\x)});
\draw [red,thick,domain=-55:55] plot ({3.5*cos(\x)}, {3.5*sin(\x)});
\draw [red,domain=-55:55] plot ({-3.5*cos(\x)+4.015}, {3.5*sin(\x)});

\node (A1) at (-0.5,0) [circle,dashed,draw] {};
\node (A2) at (0,1.5) [circle,dashed,draw] {};
\node (A3) at (-2,2) [circle,dashed,draw] {};
\node (A4) at (-2.5,0.5) [circle,dashed,draw] {};
\node (A5) at (-1.5,-2.5) [circle,dashed,draw] {};
\node (A6) at (0,-1.5) [circle,dashed,draw] {};
\node (A7) at (2,-2) [circle,dashed,draw] {};
\node (A8) at (2,2) [circle,dashed,draw] {};
\node (A9) at (2.5,0) [circle,dashed,draw] {};
\node (A10) at (-2.5,-1.2) [circle,dashed,draw] {};

\node (sledz) at (4,0.5)[right,red] {Region $A$};
\node (sledz) at (4,-0.5)[right,red] {$\mathcal{H}_{A}$};
\node (sledz) at (-4,0.5)[left] {Region $B$};
\node (sledz) at (-4,-0.5)[left] {$\mathcal{H}_{B}$};

\node (sledz) at (0.5,0)[right,red] {$\sigma_{\text{min}}$};
\node (sledz) at (0,4)[] {boundary};
\node (sledz) at (0,2.2)[dashed] {bulk};

\draw[dashed] (A5)--(A10)--(A4)--(A1) -- (A2)-- (A3);
\draw[dashed] (A3) -- (A4);
\draw[dashed] (A1)--(A6);
\draw[dashed,red] (A6)--(A7);
\draw[dashed] (A7) -- (A9) -- (A8); 
\draw[dashed,red] (A2) -- (A8);
\draw[dashed]  (A6) -- (A5);
\draw[dashed] (A3) -- (-3,2.5); 
\draw[dashed] (A5) -- (-2.4,-3.4); 
\draw[dashed] (A8) -- (3,2.5);  
\draw[dashed] (A7) -- (3,-2.5);
\draw[dashed] (A9) -- (4,0);
\draw[dashed] (A10) -- (-3.8,-1.5);
\end{tikzpicture}
\caption{Duality between a boundary theory to a theory in the bulk. The Ryu-Takayanagi proposal states a proportionality between $S(\rho_A)$ and $\text{Area}(\sigma_{\text{min}})$. In the dashed network model, $\text{Area}(\sigma_{\text{min}})$ corresponds to the number of dual links.}
\label{FIGrt}
\end{figure}

After introducing the pairwise entanglement projections to construct network states $\ket{\Phi_{\Gamma}}$ in the previous Section, we are now interested in the resulting global entanglement structure of such networks, which we will quantify by the R\'{e}nyi entanglement entropy. In a previous work \cite{13Chi}, building on the established dictionary between GFT states and (generalized) random tensor networks, some of the authors have computed the R\'{e}nyi entropy and derived the Ryu-Takayanagi entropy formula by using a simple approximation to a complete definition of a random tensor network evaluation seen as a GFT correlation function, along the lines given in \cite{Hayden}. Such a derivation was limited to the case of a non-interacting GFT model, leaving open the question about the effect of the interactions on the holographic scaling of the entropy. To answer this question, here we perform the calculation of the entanglement entropy for a simple interacting group field theory model, corresponding to Boulatov's model for topological $3d$ gravity\footnote{It must be noted, however, that such model has limited gravitational features, since it corresponds to topological gravity.} \cite{Boulatov,Ooguri}. From the statistical-mechanics point of view, the interacting model realizes a general non-Gaussian probability distribution over random tensor networks.

We first shortly review the main setting of the derivation of the Ryu-Takayanagi entropy for the group field network state $\ket{\Phi_{\Gamma}}$. Thereby, we focus on the effects of the GFT interaction terms combinatorics in the derivation of the leading contributions to the entanglement entropy. We then provide a set of theorems aiming at a classification of the different interaction combinatoric patterns in the calculation of the Feymann diagrams divergences of the perturbative expansion of the GFT partition function.

\subsection{Replica Trick and Entanglement Statistics}

Let us consider the group field network state $\ket{\Phi_{\Gamma}}$ defined on a $d$-valent graph $\Gamma=(V,E\cup \partial \Gamma)$. To a given partition $A\cup B=\partial \Gamma$ of the open legs we associate a factorization $\mathcal{H}_{\partial \Gamma}=\mathcal{H}_A \otimes \mathcal{H}_B$. Given a network state $\ket{\Phi_{\Gamma}}$, and a density matrix $\rho=\ket{\Phi_{\Gamma}}\bra{\Phi_{\Gamma}}$, the $N$th R\`enyi entanglement entropy of the reduced density matrix $\rho_A:=\tr_B[\rho]$ is defined by
\begin{align}
S_N(A)=\frac{1}{N-1}\ln \left[\frac{\tr_B[\rho_A^{N}]}{\tr_{\partial \Gamma}[\rho]^N}\right]
\label{renyi}
\end{align}

In the limit $N\rightarrow1$, for positive real numbers $N$, the function $S_N(A)$ reduces to the entanglement entropy $S(A)$, which is the von-Neumann entropy of $\rho_A$. Due to the linear character in the arguments, the R\`enyi formula in \eqref{renyi} simplifies the computation of the entanglement. 
Exponentiation of equation (\ref{renyi}) results in
\begin{align}
\label{Renyiexponented}
e^{(1-N)S_N(A)}=\frac{tr_{\mathbb{H}_{(A)}}[\rho_A^N]}{tr_{\mathbb{H}}[\rho]^N}=:\frac{Z^{(N)}_A}{Z^{(N)}_0}
\end{align}
and the calculation of $S_N(A)$ reduces to the computation of the two partition functions

\begin{align}
Z_A^{(N)}:=\tr[\rho_A^{N}]&=\tr_{A}\left[ (\tr_{B\cup E}[\bigotimes\limits_{e\in E} \rho^{(e)} \bigotimes\limits_{v \in V} \rho^{(v)}])^{N} \right]\\
Z_0^{(N)}:=\tr[\rho]^N&=\tr_{A\cup B \cup E}\left[ \bigotimes\limits_{e\in E} \rho^{(e)} \bigotimes\limits_{v \in V} \rho^{(v)} \right]^N
\end{align}
where in $Z_A^{(N)}$ a trace over region $A$ is  performed, after the $N$-times composition of the reduced density $\rho_A$. 

The effect of the trace $\tr_{B\cup E}$ in $Z_A^{(N)}$ can be represented by the action of a swap operator $\mathbb{F}$ permuting the order of the leg space $\mathcal{H}_{\partial A}$ as the cyclic element in the permutation group $\mathcal{S}_N$ \cite{13Chi}. This is apparent by closing the open links in $\{\partial \Gamma\}$ by virtual single valent vertices $w$, which are decorated with the density matrices $\rho^{(w)}=\mathbb{I}$ (the identity) and $\rho^{(w)}=\mathbb{F}$ acting on $\mathcal{H}_{\partial \Gamma}$. 
\begin{align}
\label{defZA}
Z_A^{(N)}&=\tr_{A\cup B \cup E} \left[ \bigotimes\limits_{e\in E\cup \partial \Gamma} (\rho^{(e)})^{\otimes N} \bigotimes\limits_{v \in V} (\rho^{(v)})^{\otimes N} \bigotimes\limits_{w \in A} \mathbb{F} \bigotimes\limits_{w \in B} \mathbb{I}  \right]\\
\label{defZ0}
Z_0^{(N)}&=\tr_{A\cup B \cup E} \left[ \bigotimes\limits_{e\in E\cup \partial \Gamma} (\rho^{(e)})^{\otimes N} \bigotimes\limits_{v \in V} (\rho^{(v)})^{\otimes N} \bigotimes\limits_{w \in \partial \Gamma} \mathbb{I}  \right]
\end{align}
The structure of these variables can be effectively captured by a stranded graphs representation (figure \ref{FIGfullvariable}), where each strand represents the leg of a tensor. The $N$-times tensor product of the node densities in $V$ and $\partial\Gamma$ is represented by $N$ incoming copies, corresponding to the quanta in $\mathcal{H}^{\otimes d}$, and by $N$ outgoing copies, corresponding to the dual states in $\mathcal{H}^{\otimes d}$. Performing the trace over all leg spaces after composition with the link densities corresponds to the contraction of the legs on the incoming and outgoing copies and is represented by connected strands.

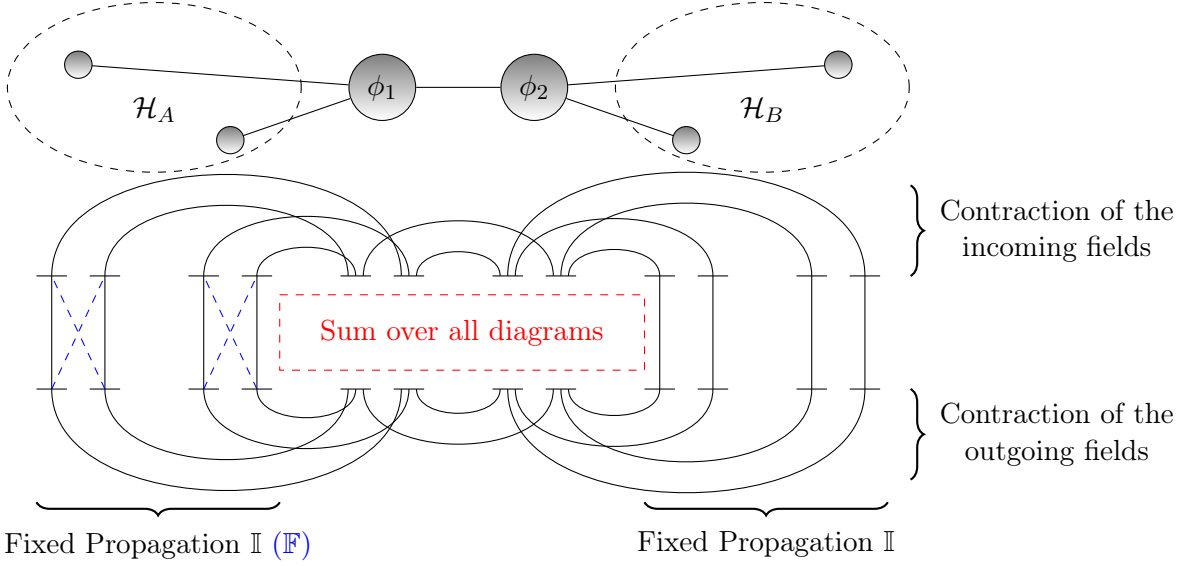
\begin{figure}
\centering
\begin{tikzpicture}
\foreach \i in {0,1,2,3,4,5}
{
\draw[-] (2*\i,0) -- (2*\i+0.4,0) ;
\draw[-] (2*\i+0.7,0) -- (2*\i+1.1,0) ;
\draw[-] (2*\i+0.7,1.5) -- (2*\i+1.1,1.5) ;
\draw[-] (2*\i,1.5) -- (2*\i+0.4,1.5) ;
}

\foreach \x in {0,1}
{
\draw[-] (2*\x+0.2,0) -- (2*\x+0.2,1.5) ;
\draw[-] (2*\x+0.9,0) -- (2*\x+0.9,1.5) ;
\draw[dashed,blue] (2*\x+0.2,0) -- (2*\x+0.9,1.5) ;
\draw[dashed,blue] (2*\x+0.9,0) -- (2*\x+0.2,1.5) ;

\draw[-] (2*\x+8.2,0) -- (2*\x+8.2,1.5) ;
\draw[-] (2*\x+8.9,0) -- (2*\x+8.9,1.5) ;
\draw (0.7-0.7*\x+0.2,0) to[bend left=-90] (0.7*\x+0.1+4,0);
\draw (0.7-0.7*\x+2.2,0) to[bend left=-90] (0.7*\x+0.2+4,0);
\draw (0.7-0.7*\x+0.2,1.5) to[bend left=90] (0.7*\x+0.1+4,1.5);
\draw (0.7-0.7*\x+2.2,1.5) to[bend left=90] (0.7*\x+0.2+4,1.5);
\draw (8.7-0.7*\x+0.2,0) to[bend left=90] (0.7*\x+2.3+4,0);
\draw (8.7-0.7*\x+2.2,0) to[bend left=90] (0.7*\x+2.2+4,0);
\draw (8.7-0.7*\x+0.2,1.5) to[bend left=-90] (0.7*\x+2.3+4,1.5);
\draw (8.7-0.7*\x+2.2,1.5) to[bend left=-90] (0.7*\x+2.2+4,1.5);

\draw (4.7-0.7*\x+0.3,1.5) to[bend left=90] (0.7*\x+2.1+4,1.5);
\draw (4.7-0.7*\x+0.3,0) to[bend left=-90] (0.7*\x+2.1+4,0);
}

\draw[dashed, color=red] (3.2,1.25)--(8,1.25)--(8,0.25)--(3.2,0.25)--(3.2,1.25);
\node (A) at (5.6,0.75) [red] {Sum over all diagrams};

 \draw [thick,decorate,decoration={brace,amplitude=5pt}] (3.2,-1.5) -- (0,-1.5);
 \node (C) at (1.6,-1.75) [below] {Fixed Propagation $\mathbb{I}$ \textcolor{blue}{($\mathbb{F}$)}};
 
 
 \draw [thick,decorate,decoration={brace,amplitude=5pt}] (11.2,-1.5) -- (8,-1.5);
 \node (C) at (9.6,-1.75) [below] {Fixed Propagation $\mathbb{I}$};

 \draw [thick,decorate,decoration={brace,amplitude=5pt}] (11.5,2.7) -- (11.5,1.5);
  \node (C) at (11.75,2.1) [right] {\Umbruch{Contraction of the incoming fields}};
  
 \draw [thick,decorate,decoration={brace,amplitude=5pt}] (11.5,0) -- (11.5,-1.2);
 \node (C) at (11.75,-0.6) [right] {\Umbruch{Contraction of the outgoing fields}};
 
\node (A) at (4.55,4) [circle,shade,draw] {$\phi_1$};
\node (B) at (6.55,4) [circle,shade,draw] {$\phi_2$};
\draw (A)--(B);
\draw (A)--(2.55,3.3)node [circle,shade,draw] {};
\draw (A)--(0.55,4.3)node [circle,shade,draw] {};
\draw (B)--(8.55,3.3)node [circle,shade,draw] {};
\draw (B)--(10.55,4.3)node [circle,shade,draw] {};

\draw[dashed] (1.55,4) ellipse (55pt and 32pt) node[below] {$\mathcal{H}_{A}$};  
\draw[dashed] (9.55,4) ellipse (55pt and 32pt) node[below] {$\mathcal{H}_{B}$}; 
\end{tikzpicture}
\caption{Stranded graph representation of the variables $Z_0^{(2)}$ (Propagation $\mathbb{I}$ in $A$ and $B$) and $Z_A^{(2)}$ (Propagation $\mathbb{F}$ in $A$ and $\mathbb{I}$ in $B$). Any stranded diagram contributing to the $2N|V|$-point function produces a contribution to the variables by the fixed contraction scheme after insertion into the red box.}
\label{FIGfullvariable}
\end{figure} 

\subsection{Calculation setting for an interacting GFT Model}


The group field probability distribution $d\nu(\phi)$ has support on the space of integrable fields $L^2(G^{\times d},\mu^{\times d})$ with respect to the Haar measure $\mu$ on $G$. Taking this space as product of leg spaces $\mathcal{H}$, the space of node tensors $\ket{\phi^{(v)}}$ is given by $\mathcal{H}^{\otimes d}:=L^2(G^{\times d},\mu^{\times d})$. In case of finite groups $G$, the space $L^2(G^{\times d},\mu^{\times d})$ has finite dimensions and the theory reduces to random tensors, as we have discussed in Section \ref{SECRandomTN}. 

For the calculations in the following Sections we need to specify the probability measure $d\nu(\phi)$ of Haar-integrable group fields $\phi$ by its weight. With respect to the uniform measure, the weight is the exponential of the negative action $S(\phi)$, which consists in the perturbation of a Gaussian weight $S_0$ with the term $\lambda S_{\text{int}}$, as described in \eqref{DEFperturbation}. For any possible choice of the kernels $\mathcal{K}$, we further demand the symmetry condition
\begin{align}
\mathcal{K}(g_1,...,g_d,\overline{g}_1,...,\overline{g}_d)=\mathcal{K}(h\circ g_1,...,h\circ g_d,\overline{h}\circ\overline{g}_1,...,\overline{h} \circ \overline{g}_d) \quad \forall h,\tilde{h}\in G\,,
\label{symmetryK}
\end{align}
This choice, corresponding to enforcing a kinematical \emph{closure constraint} at the nodes, is the minimal condition to put the networks dynamics in relation to topological 3$d$-gravity models \cite{Rovelli,Thiemann,01Ash,02Per}. Because the symmetry (\ref{symmetryK}) prevents the kernel $\mathcal{K}$ from being invertible, we need to restrict the probability measure (\ref{DEFfreemeasure}) to the space of $G$-symmetric group fields, which were discussed in Section \ref{SECsymmetricgroupfields}. However, we can safely treat the probability measures as defined on the full space $L^2(G^{\times d},\mu^{\times d})$ and ensure the symmetry constraint by projections, which are implemented by averages (\ref{symmetryprojection}) over the gauge parameters $h$.

Notice that while a simple kernel (\ref{deltakernel}) is constructed by delta functions between the arguments $g$ and $\overline{g}$ of the same index, the symmetry (\ref{symmetryK}) is enforced by group averaging due to the invariance of the Haar measure $\mu$:
\begin{align}
\mathcal{K}^{\text{sym}}(g_1,...,g_d,\overline{g}_1,...,\overline{g}_d)=\int d\mu(h) \, \prod_{i=1}^d \delta(h\circ g_i\circ \overline{g}_i^{-1})
\label{choiceK}
\end{align}
The kernel (\ref{choiceK}) corresponds graphically to a collection of strands, each representing a delta function between incoming arguments $g$ and outgoing arguments $\overline{g}$ (figure \ref{FIGsymkernels}). 
The interaction kernel $\mathcal{V}$ is defined in an analogous way, a specific form for the rank-$3$ group field theory under study is given in Section \ref{SECint}.\\

\begin{figure}[t]
\centering
\begin{tikzpicture}[scale=0.5]
\begin{scope}[shift={(-5,0.3)},rotate=90]
\draw[-] (0,0) node[below] {$\overline{\phi}$} -- (5,0) node[above] {$\phi$} ;
\draw[-] (0,0.5) -- (5,0.5);
\draw[-] (0,-0.5) -- (5,-0.5); 
\draw[-] (2.5,0) ellipse (8pt and 30pt) ; 
\end{scope}

\node (X) at (-7,5.5){$\mathcal{K}:$};
\node (X) at (0,5.5){$\mathcal{V}:$};

\begin{scope}[shift={(6,0.6)},scale=0.8,rotate=45]
\draw(-3,4) arc (270:360:3);
\draw(1,7) arc (180:270:3);
\draw(0,0) arc (0:90:3);
\draw[-] (0.5,0) node[right] {} -- (0.5,7) node[left] {};
\draw[-] (0.5,1) ellipse (30pt and 8pt); 
\draw[-] (0.5,6) ellipse (30pt and 8pt); 

\draw[-] (-3,3.5) node[left] {} -- (4,3.5) node[right] {} ;
\draw[-] (-2,3.5) ellipse (8pt and 30pt); 
\draw[-] (3,3.5) ellipse (8pt and 30pt); 
\draw(4,3) arc (90:180:3);

\node (A) at (-3.3,3.2) [left] {$\phi_2$};
\node (A) at (0.3,-0.2) [right] {$\overline{\phi}_2$};
\node (A) at (0.8,7.3) [left] {$\phi_1$};
\node (A) at (4.3,3.8) [right] {$\overline{\phi}_1$};
\end{scope}
\end{tikzpicture}
\caption{Stranded graph representation of the propagation kernel $\mathcal{K}$ in \eqref{choiceK} and the interaction kernel $\mathcal{V}$ in \eqref{DEFv} for the case of a rank $d=3$ group field. Strands represent delta functions and ellipses gauge parameter $h$ shifting the arguments. The group element $h$ enforcing the symmetry of $\mathcal{K}$ appears in all delta functions and its graphically denoted by an ellipse crossing the strands. }
\label{FIGsymkernels}
\end{figure}
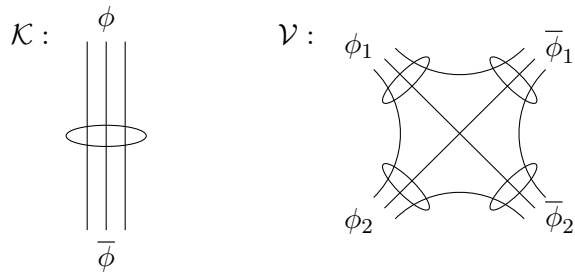

Now, the field-theoretic description allows us to describe the entropy $S_N(A)$, as well as the partition functions $Z_0^{(N)}$ and $Z_A^{(N)}$, as random variables dependent on the field configuration $\phi^{(v)}$, taken to be identical for all vertices $v$. In particular, given the field-theoretic random character, we expect the fluctuations of $S_N(A)$ around its average $\mathbb{E}[S_N(A)]$ to be exponentially suppressed in the limit of high bond dimension of the leg spaces, as a consequence of the phenomenon of measure concentration, as suggested in \cite{HayEnt}. This has an important impact on the derivation of the entanglement entropy, which effectively reduces to a computation of the expectation value $\mathbb{E}[S_N(A)]$. 

Analogously, the variables $Z_0^{(N)}$ and $Z_A^{(N)}$ shall concentrate around their averages, allowing for a further approximation of $\mathbb{E}[S_N(A)]$ in terms of the individual averages $\mathbb{E}[Z_0^{(N)}]$ and $\mathbb{E}[Z_A^{(N)}]$ \cite{Hayden},
\begin{align}
\mathbb{E}[S_N(A)]\approx\frac{1}{N-1}\ln\left[\frac{\mathbb{E}[Z_A^{(N)}]}{\mathbb{E}[Z_0^{(N)}]}\right]
\label{centralapproximation}
\end{align}
which are now defined in terms of the perturbed Gaussian measure for the interacting GFT, \begin{align}
\mathbb{E}[Z_{A/0}^{(N)}]= \int_{L^2(G^{\times d})} d\nu(\phi)\, Z_{A/0}^{(N)}[\phi]=\int_{L^2(G^{\times d})} [D\phi]\, Z_{A/0}^{(N)}[\phi] \, e^{-\left[S_0(\phi)+\lambda S_{\text{int}}(\phi)\right]}\quad.
\label{expectation}
\end{align}

By taking the perturbative expansion of $\mathbb{E}[Z_{A/0}^{(N)}]$ in orders of $\lambda$ and applying Wicks theorem for Gaussian random variables, the expectation (\ref{expectation}) decomposes into a sum of associated Feynman diagrams contributions. Due to the linearity of the trace, as sketched in figure \ref{FIGfullvariable}, the integration over the group fields can be carried out independently before the contraction with the links. The expectations $\mathbb{E}[Z_{A/0}^{(N)}]$ define $2N|V|$-point functions of the group field theory, which are then contracted by a pattern determined by the network geometry. The red box in figure \ref{FIGfullvariable} represents the sum of all Feynman diagrams of the $2N|V|$-point functions, each of which corresponds to a fully contracted diagram. 

Each stranded Fenyman diagram $\mathcal{G}$ consists of edges $l$ representing field propagations, which are bundles of parallel strands, and faces $f$ defined by closed strands with boundary $\partial f$, given by an oriented and ordered collection of edges. A holonomy $h_l$ is associated with each edge $l$, corresponding to the group element enforcing the symmetry constraint. Each face $f$ contributes with a delta function of its holonomy and the amplitude of a diagram $\mathcal{G}$ is given by:

\begin{align}
\mathcal{A}[\mathcal{G}]=\int [\prod\limits_{l\in \mathcal{G}} dh_l] \prod\limits_{f \in \mathcal{G}} \delta(\prod_{l\in \partial f}h_l)
\label{amplitude}
\end{align}

From the infinite number of possible diagrams $\mathcal{G}$ representing contracted propagation and interaction processes, we want to identify the processes with dominant amplitudes in the limit of high dimensions $D$ of the leg spaces $L^{2}(G,\mu)$, which we generically take as finite-dimensional through of a sharp cut-off $\Lambda$ in the group representation, such that
\footnote{A stronger restriction consists in restricting to finite groups $G$ with $D$ elements. In this case, the Dirac distribution $\delta$ can be understood as an element in $L^2(G,\mu)$ and it holds $\delta(e)=D$ for the identity $e$ in $G$. Nevertheless, the group should remain non-abelian. More radically, one could generalise the derivation and regularize the divergences via ``box'' normalization of $\delta_g \in L^2[G, \d\mu]$ by using quantum groups. As shown in \cite{MAJOR1996267, PhysRevD.84.064010}, the quantum deformation relates to the cosmological constant $\Lambda$ in the semi-classical regime of the spinfoam formalism.
Interestingly, the cosmological constant $\Lambda$ in the link space dimension \eqref{bond} would make our vacuum state a dS vacuum if $\Lambda>0$ and AdS vacuum if $\Lambda<0$.} 
\begin{equation}\label{bond}\delta(e) = D(\Lambda)\end{equation}.
  
Therefore, the divergence degree $\Omega$ of $\mathcal{G}$ is defined as the exponent of
\begin{align}
\mathcal{A}[\mathcal{G}]=:\delta(e)^{\Omega[\mathcal{G}]}
\end{align}
The dominant contributions to the averages (\ref{expectation}) are the maxima of $\Omega[\mathcal{G}]$. Their number and the maximal divergence degree determine the asymptotic behavior of the expectations such as (\ref{expectation}) and are thus of central interest in the derivation of the expected entanglement entropy (\ref{centralapproximation}). The following Sections are dedicated to identifying the maximal divergence degrees in the different settings of free and perturbed group field theory.


\section{Maximal Divergent Contributions: A General Scheme}
Diagrams $\mathcal{G}$ representing contributions to the expectation values of $Z_{A/0}^{(N)}$ can have various shapes, but only the maximal in the bond dimension $D$ divergent contributions are relevant in the calculation of entanglement entropies. Bounds of the divergence degree are given by the number of faces of each diagram, which motivate us to study the maximum number of faces for the two observables $Z_0^{(N)}$ and $Z_A^{(N)}$. In the free theory, we will find the maximal face numbers in a class of diagrams that we call \emph{locally averaged diagrams} and restrict our search for the maximal divergence degree afterwards to this class. 

\subsection{Local Processes by Independent Node Averaging}
\label{SEClocalevidence}

Let us first consider the case of a small number of interactions in a diagram $\mathcal{G}$, corresponding to a term in a small order of $\lambda$ in the perturbative expansion (\ref{intexpension}) of the expectation value $\mathbb{E}[Z_{0/A}^{(N)}]$. In this case, the stranded structure sketched in figure \ref{FIGfullvariable} is dominated by the contraction scheme determined by the network topology. In particular, the number of edges carrying gauge parameters integrated in the amplitude \eqref{amplitude} is proportional to the number of network nodes. This allows for an estimation of the number of parameter evaluation needed to fix all face holonomies, suggesting a correlation between the number of faces, which we will denote by $\widetilde{\Omega}[\mathcal{G}]$, and the divergence degree $\Omega[\mathcal{G}]$. We shall use this intuition to single out a class of diagram, where the divergence $\Omega$ is expected to be maximal. 

Due to their different structure at the boundary $\partial \Gamma$ (figure \ref{FIGfullvariable}), we distinguish between the variables $Z_0^{(N)}$ and $Z_A^{(N)}$ and proof optimality statements of different generality. A central aspect is the notation of locality, which in our context is understood as propagation processes happening just between field copies associated with the same network vertex $v$. Such propagation processes are indexed by the permutation group $S_N$ as introduced in (\ref{permutationdecompositon}), where the identity $\mathbb{I}$ and the cyclic permutation $\mathbb{F}$ are of particular interest.

\begin{theorem}
Let us assume a network graph $\Gamma=(V,E\cup \partial\Gamma)$, such that every node is path connected to boundary links. The only diagram contributing to $\mathbb{E}[Z_0^{(N)}]$ in the free sector, which maximizes the face number $\widetilde{\Omega}$, is given by local propagation with the identity $\mathbb{I}$ at each node. 
\label{THEconnected}
\end{theorem}

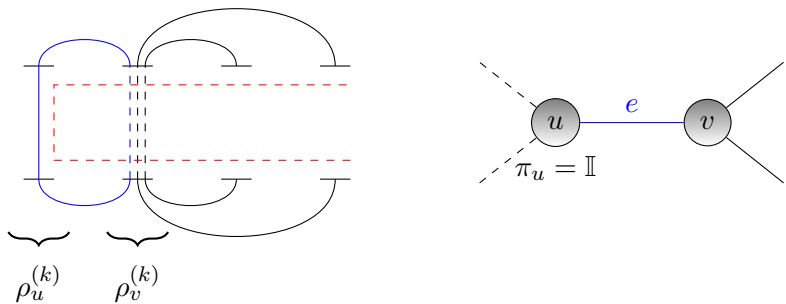
\begin{figure}[t]
\centering
\begin{tikzpicture}
\foreach \i in {0,1,2,3}
{
\draw[-] (1.3*\i,0) -- (1.3*\i+0.4,0) ;
\draw[-] (1.3*\i,1.5) -- (1.3*\i+0.4,1.5) ;
}

\draw[blue] (0.2,0)--(0.2,1.5);
\draw[blue,dashed] (1.4,0)--(1.4,1.5);
\draw[dashed] (1.5,0)--(1.5,1.5);
\draw[dashed] (1.6,0)--(1.6,1.5);

\draw[blue] (0.2,0) to[bend left=-90] (1.4,0);
\draw (2.8,0) to[bend left=90] (1.6,0);
\draw (4.1,0) to[bend left=90] (1.5,0);
\draw[blue] (0.2,1.5) to[bend left=90] (1.4,1.5);
\draw (2.8,1.5) to[bend left=-90] (1.6,1.5);
\draw (4.1,1.5) to[bend left=-90] (1.5,1.5);
\draw[dashed, color=red] (4.3,0.25)--(0.4,0.25)--(0.4,1.25)--(4.3,1.25);
 \draw [thick,decorate,decoration={brace,amplitude=5pt}] (0.6,-0.7) -- (-0.2,-0.7); 
\draw [thick,decorate,decoration={brace,amplitude=5pt}] (1.9,-0.7) -- (1.1,-0.7);
\node (C) at (0.2,-1) [below] {$\rho_u^{(k)}$} ;
\node (C) at (1.5,-1) [below] {$\rho_v^{(k)}$} ;

\begin{scope}[shift={(9,0.75)}]
\node (A) at (0,0) [circle,shade,draw] {$v$};
\node (B) at (-2,0) [circle,shade,draw] {$u$};
\node (C) at (-2,-0.6)[]{$\pi_u=\mathbb{I}$};
\draw[blue] (A)--(B) node[midway,above]{$e$};
\draw[dashed] (-3,-0.8)--(B)--(-3,0.8);
\draw (1,0.8)--(A)--(1,-0.8);
\end{scope}
\end{tikzpicture}
\caption{Sketch of the induction argument, exploited in the proof of theorem \ref{THEconnected}: If in the $k$th network copy a local propagation $\mathbb{I}$ is happending at a neighboring node $u$, the same process must happen at $v$ for the blue strand to close directly.}
\label{FIGiterationargumentz0}
\end{figure}

\begin{proof}
Let us treat the open links of the network as incident to single-valent boundary nodes with fixed local propagation determined by the boundary conditions, thus in the case of $Z_0^{(N)}$ by the symbol $\mathbb{I}$. Each face of a diagram in the free theory includes at least one contraction of the incoming fields, since each outgoing field propagates to an incoming field. If thus each face includes just one incoming contraction, as it is the case for the pattern $\{\mathbb{I}\}$, the divergence degree is maximal.
We prove the uniqueness of this maximal divergent case by induction through the network starting from its boundary with fixed propagation $\mathbb{I}$. Let us assume the fields of a neighbor $v$ to the node $u$ are propagating locally with the symbol $\mathbb{I}$, as sketched in figure \ref{FIGiterationargumentz0}. For the $k$th copy of the link $e$ between $v$ and $u$ to be the single incoming contraction of the associated face, the $k$th copy of the incoming and outgoing field of the node $u$ must propagate into each other, as sketched dashed in figure \ref{FIGiterationargumentz0}. Applying this argument to all $N$ copies of the link, the field copies of the node $u$ have to propagate locally with the symbol $\mathbb{I}$, such that the link to $v$ contributes maximally to the face number. Since the network is path connected to the boundary, the induction reaches all nodes.
\end{proof}

In contrast to the homogeneous boundary situation of the variable $Z_0^{(N)}$, the boundary regions $A$ and $B$ are differently treated in the variable $Z_A^{(N)}$. By assuming the boundary regions to be connected by the bulk network, we will always find closed strand including more than one contraction of incoming fields. To proof a similar result to theorem \ref{THEconnected}, in this boundary situation, we need to make further assumptions on the network graph.

\begin{theorem}
\label{THEtreenetwork}
Let us assume a tree graph $\Gamma=(V,E\cup \partial\Gamma)$ and a partition $A\cup B = \partial\Gamma$, such that we find a link $e\in E$ separating regions connected to $A$ and $B$. Then all diagrams contributing to $\mathbb{E}[Z_{A}^{(N)}]$ in the free sector, which maximize the face number $\widetilde{\Omega}$, are included in the local averages. 
\end{theorem}

\begin{proof}
After omission of the link $e$ we have two trees, the first connected just to the region $A$ and the second to $B$. By the minimal path length $k$ to the root node, taken to be the node incident to $e$, we classify each node of both trees into a layer. Assuming a diagram maximizing $\widetilde{\Omega}$ we will now proof in both trees the local propagation by induction from the deepest layer to the layers with smaller indices. As in the proof of theorem \ref{THEconnected}, the induction starts with boundary links, which are treated as additional network nodes with fixed propagation by $\mathbb{I}$ in the tree to $A$ and $\mathbb{F}$ to $B$.

Let us thus assume a node $v$ in the $k$th layer, which has two children $u_1$ and $u_2$ in the $(k$+$1)$th layer with the same local propagation happening w.l.o.g. by $\mathbb{I}$ (figure \ref{FIGiterationargument} b). Let us further assume, that an incoming copy of the node field is involved in a nonlocal propagation, as sketched in blue. We find a nonlocal propagating outgoing copy to the nonlocal propagating incoming copy, such that the strands associated with the links from $v$ to $u_1$ and $u_2$ connect both nonlocal propagating copies of the node field $v$\footnote{This can be done by simply following two strands, which need to close in each situation.}. We now modify the process by a local propagation between the identified pair of nonlocal propagating incoming and outgoing fields, sketched in figure \ref{FIGiterationargument} b) by dark dashed lines. Due to the special choice of the nonlocal propagating pair at the node $v$, the number of strands with the first two indices increases by two. This compensates the maximal decrease of one in the number of closed strands associated with the third index. The modified diagram has thus a higher number of faces, thus the assumption of a nonlocal propagation contradicts the assumption of maximal face number $\widetilde{\Omega}$. 

Induction with decreasing layer number $k$ reaches both subtrees of the network, thus at each node a local process has to happen to maximize $\widetilde{\Omega}$.
\end{proof}

\begin{figure}[t]
\centering
\begin{tikzpicture}
\node (Z) at (-1,3) {b)};

\begin{scope}[shift={(-6,1.5)}]
\node (Y) at (0,1.5) {a)};

\node (A) at (1,1) [circle,shade,draw] {};
\node (B) at (0.4,0) [circle,shade,draw] {};
\node (C) at (1.6,0) [circle,shade,draw] {};
\draw (B)--(A)--(C);
\draw (0.8,-1) -- (B)--(0,-1);
\draw (1.2,-1) -- (C)--(2,-1);


\begin{scope}[shift={(2.5,0)}]
\node (A1) at (1,1) [circle,shade,draw] {};
\node (B1) at (0.4,0) [circle,shade,draw] {};
\node (C1) at (1.6,0) [circle,shade,draw] {};
\draw (B1)--(A1)--(C1);
\draw (0.8,-1) -- (B1)--(0,-1);
\draw (1.2,-1) -- (C1)--(2,-1);
\end{scope}
\draw (A)--(A1);

\draw[dashed] (3.5,-1) ellipse (35pt and 15pt); 
\node (X) at (3.5,-1.2) [] {$\mathcal{H}_{B}$};

\draw[dashed] (1,-1) ellipse (35pt and 15pt); 
\node (X) at (1,-1.2) [] {$\mathcal{H}_{A}$};
\end{scope}

\foreach \i in {0,1,2,3}
{
\draw[-] (1.3*\i,0) -- (1.3*\i+0.4,0) ;
\draw[-] (1.3*\i,1.5) -- (1.3*\i+0.4,1.5) ;
}

\draw (0.2,0)--(0.2,1.5);
\draw (1.5,0)--(1.5,1.5);

\draw (0.2,1.5) to[bend left=90] (2.8,1.5);
\draw (1.5,1.5) to[bend left=90] (2.7,1.5);
\draw (2.9,1.5) to[bend left=90] (4.1,1.5);

\draw (0.2,0) to[bend left=-90] (2.8,0);
\draw (1.5,0) to[bend left=-90] (2.7,0);
\draw (2.9,0) to[bend left=-90] (4.1,0);

\draw[dashed, color=red] (6.4,0.25)--(1.7,0.25)--(1.7,1.25)--(6.4,1.25);

\draw[-] (5.6,0) -- (6,0) ;
\draw[-] (6,1.5) -- (6.4,1.5) ;

\draw (6.1,1.5) to[bend left=-60] (5.8,1.8);
\draw (6.2,1.5) to[bend left=30] (6.4,1.9);
\draw (6.3,1.5) to[bend left=60] (6.6,1.8);

\draw (5.7,0) to[bend left=60] (5.4,-0.3);
\draw (5.8,0) to[bend left=-30] (6,-0.4);
\draw (5.9,0) to[bend left=-60] (6.2,-0.3);

\draw[blue,dashed] (2.8,1.5) -- (5.8,0);
\draw[blue,dashed] (2.8,0) -- (6.2,1.5);
\draw[dashed] (2.8,1.5) -- (2.8,0);
\draw[dashed] (6.2,1.5) -- (5.8,0);

\node (A) at (2.8,3) [circle,shade,draw] {$v$};
\node (B) at (1.5,2.7) [circle,shade,draw] {};
\node (B1) at (1.3,2.7) [left] {$\pi_{u_2}=\mathbb{I}$};

\node (C) at (0.2,3.3) [circle,shade,draw] {};
\node (C1) at (0,3.3) [above] {$\pi_{u_1}=\mathbb{I}$};

\node (D) at (4.1,3) [circle,shade,draw] {};

\draw (C) -- (A) -- (B) -- (A) -- (D);

\node (E) at (6.2,3.3) [circle,shade,draw] {};
\node (F) at (5.8,2.7) [circle,shade,draw] {};

\node (G) at (4.96,3) {$\cdots$};
\node (H) at (4.96,1.5) {$\cdots$};
\node (I) at (4.96,0) {$\cdots$};

\end{tikzpicture}
\caption{a) Network with a tree structure, such that the subtrees to $\mathcal{H}_{A}$ and $\mathcal{H}_{B}$ are connected by a single link. b) Node $v$ with two neighbors $u_1$ and $u_2$, which have fixed local propagation by $\mathbb{I}$. The blue and black dashed lines compare a nonlocal propagation with a modification to a local propagation.}
\label{FIGiterationargument}
\end{figure}
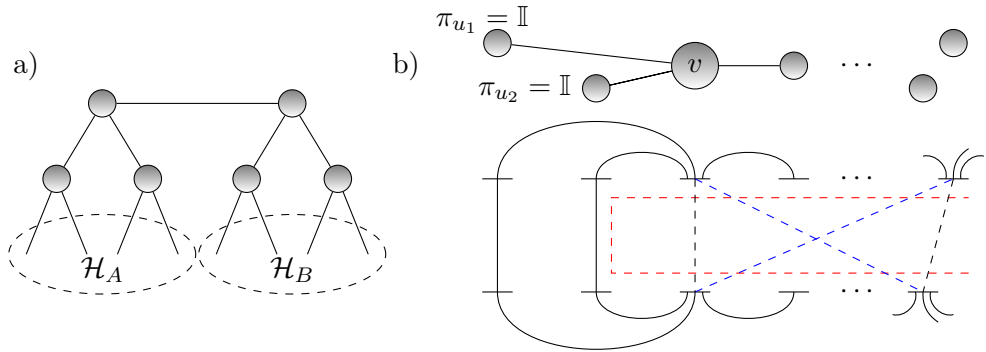


The statements concerning the unique maximum of the face number in the free theory motivate the search for maximal divergent diagrams in the class of local processes. This class corresponds to the statistics of independent vertex averaging, where we separately average the density $(\rho^{(v)})^{\otimes d}$ at each network vertex $v\in V$:

\begin{align}
\mathbb{E}[Z_{0/A}^{(N)}]\approx \tr_{E\cup\partial \Gamma} \left[ \bigotimes\limits_{e\in E\cup \partial \Gamma} (\rho^{(e)})^{\otimes N} \bigotimes\limits_{v \in V} \mathbb{E}[(\rho^{(v)})^{\otimes N}]\bigotimes\limits_{w \in B} \mathbb{I} \bigotimes\limits_{w \in A} \mathbb{I}/\mathbb{F} \right]
\end{align} 

This finite but by expectation dominant sector of possible diagrams corresponds to a different interpretation of the variables (\ref{expectation}) as composed of independent but identically distributed random fields $\phi^{(v)}$. The corresponding Feynman diagrams decompose into different subdiagrams to each vertex $v$ of the network, which are contracted by the network pattern sketched in figure \ref{FIGfullvariable}. At each vertex, we index the diagrams by an element $\pi$ of the permutation group $\mathcal{S}_N$:
\begin{align}
\mathbb{E}[(\rho^{(v)})^{\otimes N}]=\sum_{\pi \in \mathcal{S}_N} \mathbb{P}(\pi)+\mathcal{O}(\lambda) 
\label{nodeaverage}
\end{align}
The sum over the permutation group $\mathcal{S}_N$ captures all processes of the corresponding free theory ($\lambda=0$), where we define for each permutation $\pi$ an operator $\mathbb{P}(\pi)$ modelling the propagation (\ref{choiceK}) of the $k$th incoming node copy to the $\pi(k)$th outgoing. Perturbation of the free theory would give rise to diagrams with interactions, which are controled by the perturbation parameter $\lambda$.

\subsection{Maximal Divergent Diagrams for Locally Averaged Networks}
\label{SECfreelocal}

Insights into the structure of dominant patterns can be observed in the local node averages, where we first restrict to the case $\lambda=0$, therefore to the sum term in (\ref{nodeaverage}). Each diagram $\mathcal{G}$ contributing to one of the averages $\mathbb{E}[Z_{A/0}^{(N)}]$ corresponds in this sector to a choice of a permutation $\pi_v\in \mathcal{S}_N$ at each network vertex $v\in V$. 

For such permutation pattern we can categorize each face of a diagram by the unique network link $e\in E\cup \partial \Gamma$ they include, which turns the number of faces $\widetilde{\Omega}[\mathcal{G}]$ into a sum of all link contributions. Each closed strand to a link $e=(e_1,e_2)$ corresponds to a cycle in the permutation $\pi_{e_1}^{-1} \circ \pi_{e_2} \in \mathcal{S}_N $, where the difference of the cycle number $\chi(\pi_1^{-1} \circ \pi_2)$ to the maximum $N$ defines a metric $d(\pi_1,\pi_2)$ in the permutation group $\mathcal{S}_N$. The face number for a local permutation pattern $\{\pi_v\}$ as a sum of the cycle numbers along each link is thus
\begin{align}
\widetilde{\Omega}[\{\pi_v\}]=\sum_{(e_1,e_2)\in E\cup \partial \Gamma} \chi(\pi_{e_1}^{-1} \circ \pi_{e_2})=\sum_{(e_1,e_2)\in E\cup \partial \Gamma} (N-d(\pi_{e_1},\pi_{e_2}))
\end{align}

Within our aim to identify permutation pattern maximizing the divergence degree $\Omega$, we first identify those maximizing the face number $\widetilde{\Omega}$. While for $Z_0^{(N)}$ we have already found the single maximum by $\{\mathbb{I}\}$ with theorem \ref{THEconnected}, we will now build on theorem \ref{THEtreenetwork} to find the maximal face number in the situation $Z_A^{(N)}$.
 
\begin{theorem}
\label{THfaces}
Let us assume a connected network graph $\Gamma=(V,E\cup \partial\Gamma)$ and a partition $A\cup B=\partial\Gamma$. Let $\sigma_{\text{min}}\subset E$ be a link set with minimal cardinality, such that by omission of $\sigma_{\text{min}}$ the graph $\Gamma$ reduces to two connected components, the first including the boundary $A$ and the second $B$ (see figure \ref{FIGpaths}). It then holds: \\
i) A pattern maximizing the number $\widetilde{\Omega}[\{\pi_v\}]$ of faces in the situation of $Z_A^{(N)}$ is the association of $\mathbb{I}$ and $\mathbb{F}$ to the components connected to $A$ and $B$ after omission of the links $\sigma_{\text{min}}$ (figure \ref{FIGmaxfree}b). \\
ii) If $\sigma_{\text{min}}$ is unique, $\widetilde{\Omega}[\{\pi_v\}]$ has only one maximum.
\end{theorem}

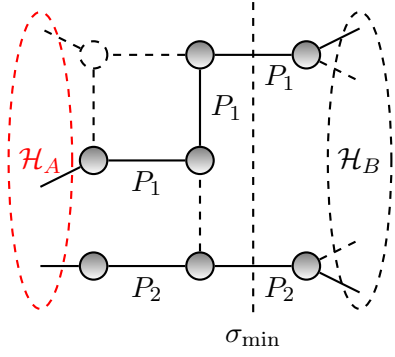
\begin{figure}
\centering
\begin{tikzpicture}[thick,scale=0.7]
\node (A1) at (0,0) [circle,shade,draw] {};
\node (A2) at (0,2) [circle,shade,draw] {};
\node (A3) at (-2,2) [circle,draw,dashed] {};
\node (A4) at (-2,0) [circle,shade,draw] {};
\node (A5) at (-2,-2) [circle,shade,draw] {};
\node (A6) at (0,-2) [circle,shade,draw] {};
\node (A7) at (2,-2) [circle,shade,draw] {};
\node (A8) at (2,2) [circle,shade,draw] {};
\draw[-] (A1) -- (A2) node[midway,right]{$P_1$};
\draw[dashed] (A2)--(A3)--(A4);
\draw[dashed] (A6)--(A1);
\draw[-] (A4) -- (A1) node[midway,below]{$P_1$};
\draw[-] (A6) -- (A5) node[midway,below]{$P_2$};
\draw[-] (A6) -- (A7);
\node (sledz) at (1.5,-2) [below]{$P_2$}; 
\draw[-] (A2) -- (A8);
\node (sledz2) at (1.5,2) [below]{$P_1$};
\draw[dashed] (A3) -- (-3,2.5); 
\draw[-] (A4) -- (-3,-0.5); 
\draw[-] (A5) -- (-3,-2); 
\draw[-] (A8) -- (3,2.5); 
\draw[-] (A7) -- (3,-2.5);
\draw[dashed] (A8) -- (3,1.5); 
\draw[dashed] (A7) -- (3,-1.5);
\draw[dashed] (1,3) -- (1,-3) node[below]{$\sigma_{\text{min}}$};
\draw[dashed,red] (-3,0) ellipse (17pt and 80pt) node {$\mathcal{H}_{A}$}; 
\draw[dashed] (3,0) ellipse (17pt and 80pt) node {$\mathcal{H}_{B}$};  

\end{tikzpicture}
\caption{Disjoint paths through a network graph $\Gamma$, such that they start in a region $A$ and end in a region $B$ of open links. As a correllary of the maximal-flow-minimal-cut theorem \cite{FordFulkerson,EliasFeinstein} the maximal number of such paths is $|\sigma_{\text{min}}|$, where $\sigma_{\text{min}}$ is a minimal set of links separating $\Gamma$ according to the chosen partition of open links.}
\label{FIGpaths}
\end{figure}

\begin{proof}
By a corollary of the maximal-flow-minimal-cut theorem \cite{FordFulkerson,EliasFeinstein}, we find a number of $|\sigma_{\text{min}}|$ disjoint paths $P_k$ through the network starting with a link in $A$ and ending in $B$ \cite{Hayden} (figure \ref{FIGpaths}). Taking just the links included in the paths into account we estimate with the triangle inequality of the metric $d$:
\begin{align}
\label{ineq1}
\widetilde{\Omega}[\{\pi_v\}]-N|E\cup\partial\Gamma|& \leq -\sum_{k=1}^{|\sigma_{\text{min}}|}\sum_{(e_1,e_2)\in P_k} d(\pi_{e_1},\pi_{e_2}) \\
& \leq -|\sigma_{\text{min}}|(N-1)
\label{ineq2}
\end{align}
Associating the trivial permutation $\mathbb{I}$ with the nodes connected to $A$ after omission of $\sigma_{\text{min}}$ and $\mathbb{F}$ to the other nodes results in a number $N|E\cup\partial\Gamma|-|\sigma_{\text{min}}|(N-1)$ of faces and is with (\ref{ineq2}) maximal, which shows i).\\
Let us assume a different pattern $\{\pi_v\}$ maximizing $\widetilde{\Omega}$. If it contains just the symbols $\mathbb{I}$ and $\mathbb{F}$, both regions would be separated by a different minimal set $\sigma_{\text{min}}$, if (\ref{ineq1}) and (\ref{ineq2}) hold straight. If other symbols are included, we set all further symbols to $\mathbb{I}$, which does not change $\widetilde{\Omega}$, if both inequalities hold straight. The resulting pattern thus also maximizes $\widetilde{\Omega}$, where the regions are separated by a minimal set $\sigma_{\text{min}}$. Setting the symbols instead to $\mathbb{F}$ would result in another minimal set $\sigma_{\text{min}}$. If $\sigma_{\text{min}}$ is unique, the pattern i) maximizing $\widetilde{\Omega}$ is also unique.
\end{proof}

\begin{figure}
\centering
\begin{tikzpicture}[thick,scale=0.7]
\node (n) at (-4,2) {a)};
\node (A1) at (0,0) [circle,shade,draw] {};
\node (A2) at (0,2) [circle,shade,draw] {};
\node (A3) at (-2,2) [circle,shade,draw] {};
\node (A4) at (-2,0) [circle,shade,draw] {};
\node (A5) at (-2,-2) [circle,shade,draw] {};
\node (A6) at (0,-2) [circle,shade,draw] {};
\node (A7) at (2,-2) [circle,shade,draw] {};
\node (A8) at (2,2) [circle,shade,draw] {};
\draw[-] (A1) -- (A2)-- (A3)-- (A4) -- (A1) -- (A6) -- (A5);
\draw[-] (A6) -- (A7); 
\draw[-] (A2) -- (A8);
\draw[-] (A3) -- (-3,2.5); 
\draw[-] (A4) -- (-3,-0.5); 
\draw[-] (A5) -- (-3,-2); 
\draw[-] (A8) -- (3,2.5); 
\draw[-] (A7) -- (3,-2.5);
\draw[-] (A8) -- (3,1.5); 
\draw[-] (A7) -- (3,-1.5);
\node (C) at (-2.75,-2.5){};
\node (B) at (2.75,-2.5){};
 \node (X) at (0,-3.25) {$\mathbb{I}$};
 \draw [thick,decorate,decoration={brace,amplitude=5pt}] (B) -- (C) node[midway,right,xshift=12pt,]{}; 
\draw[dashed] (-3,0) ellipse (17pt and 80pt) node {$\mathcal{H}_{A}$}; 
\draw[dashed] (3,0) ellipse (17pt and 80pt) node {$\mathcal{H}_{B}$};  

\begin{scope}[shift={(9,0)}]
\node (n) at (-4,2) {b)};
\node (A1) at (0,0) [circle,shade,draw] {};
\node (A2) at (0,2) [circle,shade,draw] {};
\node (A3) at (-2,2) [circle,shade,draw] {};
\node (A4) at (-2,0) [circle,shade,draw] {};
\node (A5) at (-2,-2) [circle,shade,draw] {};
\node (A6) at (0,-2) [circle,shade,draw] {};
\node (A7) at (2,-2) [circle,shade,draw] {};
\node (A8) at (2,2) [circle,shade,draw] {};
\draw[-] (A1) -- (A2)-- (A3)-- (A4) -- (A1) -- (A6) -- (A5);
\draw[-] (A6) -- (A7); 
\draw[-] (A2) -- (A8);
\draw[-] (A3) -- (-3,2.5); 
\draw[-] (A4) -- (-3,-0.5); 
\draw[-] (A5) -- (-3,-2); 
\draw[-] (A8) -- (3,2.5); 
\draw[-] (A7) -- (3,-2.5);
\draw[-] (A8) -- (3,1.5); 
\draw[-] (A7) -- (3,-1.5);
\draw[dashed] (1,3) -- (1,-3) node[below]{$\sigma_{\text{min}}$};
\node (X) at (-1,-3.25) {$\mathbb{F}$};
\node (C) at (-2.75,-2.5){};
\node (B) at (0.75,-2.5){};
 \draw [thick,decorate,decoration={brace,amplitude=5pt}] (B) -- (C) node[midway,right,xshift=12pt,]{}; 
 \node (X) at (2,-3.25) {$\mathbb{I}$};
\node (C) at (1.25,-2.5){};
\node (B) at (2.75,-2.5){};
 \draw [thick,decorate,decoration={brace,amplitude=5pt}] (B) -- (C) node[midway,right,xshift=12pt,]{}; 
\draw[dashed,red] (-3,0) ellipse (17pt and 80pt) node {$\mathcal{H}_{A}$}; 
\draw[dashed] (3,0) ellipse (17pt and 80pt) node {$\mathcal{H}_{B}$};  

\end{scope}
\end{tikzpicture}
\caption{Permutation pattern with maximal face number $\widetilde{\Omega}$ (theorem \ref{THfaces}) and divergence degree (theorem \ref{THdivergence}). a) In the boundary situation of $Z_0^{(N)}$, the maxima are achieved by association of the trivial permutation $\mathbb{I}$ and b) in the situation of $Z_A^{(N)}$ by association of $\mathbb{I}$ and $\mathbb{F}$ to regions separated by $\sigma_{\text{min}}$.}
\label{FIGmaxfree}
\end{figure}
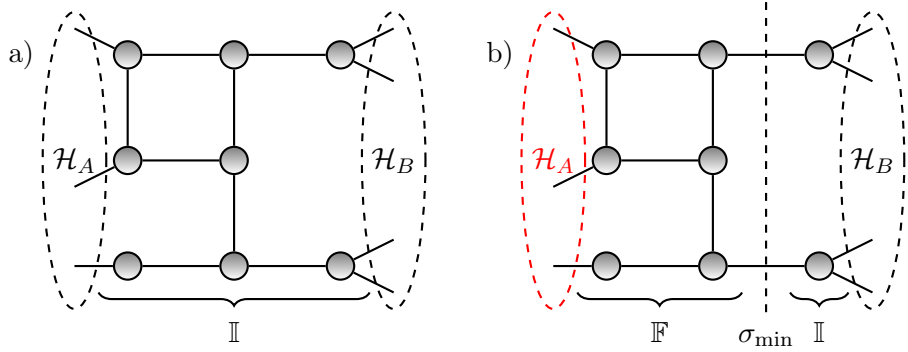

The divergence degree $\Omega[\mathcal{G}]$ of the amplitude (\ref{amplitude}) is however smaller than the face number $\widetilde{\Omega}$ due to the evaluation of the gauge integrals with parameters $h_e$ associated with each propagation. However, we can use certain bounds oriented on the face number $\widetilde{\Omega}$ to find maximal divergent patterns. We proceed by first introducing the network property of reducibility to a tree via coarse-graining, which will then be exploited to prove the uniqueness of the divergence degree maxima.

\begin{definition}
Let $\Gamma=(V,E)$ be a graph with a disjoint partition $\bigcup_m V_m = V$ of its nodes into regions $m$. The corresponding \textbf{coarse-grained graph} $\Gamma^{\{V_m\}}$ (figure \ref{FIGregion}) consists of the regions $m$ as nodes and a number $E_{m\tilde{m}}$ of links between regions $m$ and $\tilde{m}$ given by:
\begin{align}
E_{m\tilde{m}}=\#\{(e_1,e_2)\in E \,|\, e_1\in V_m,e_2\in V_{\tilde{m}}\}
\end{align}
$\Gamma$ is called \textbf{coarse-grainable to a tree}, if there exists a disjoint partition $\bigcup_m V_m= V$ of its nodes such that for all $m$ it holds $V_m\neq V$ and the corresponding coarse-grained graph $\Gamma^{\{V_m\}}$ is at most minimal connected.
\label{DEFcoarsegrainable}
\end{definition}

Network graphs $\Gamma$, however, have open links, which we will close by adding single-valent nodes on their ends. The definition of coarse-graining is thus extended to open graphs by partition of the virtual boundary nodes $\partial \Gamma$ together with the nodes $V$.

\begin{theorem}
\label{THdivergence}
With the same conditions on the network graph $\Gamma$ as in theorem \ref{THfaces} it holds:\\
i) A pattern maximizing the divergence degree $\Omega[\{\pi_v\}]$ in the boundary situation of $Z_0^{(N)}$ is the association of $\mathbb{I}$ to all nodes (figure \ref{FIGmaxfree}a).\\
ii) A pattern maximizing the divergence degree $\Omega[\{\pi_v\}]$ in the boundary situation of $Z_A^{(N)}$ is the association of $\mathbb{I}$ and $\mathbb{F}$ to the separated regions connected to $A$ and $B$ after omission of the links $\sigma_{\text{min}}$ (figure \ref{FIGmaxfree}b). \\
iii) The number of maximal divergent pattern is the same for both boundary situations $Z_0^{(N)}$ and $Z_A^{(N)}$. If $\Gamma$ is not coarse-grainable to a tree, the maxima i) and ii) are unique.
\end{theorem}

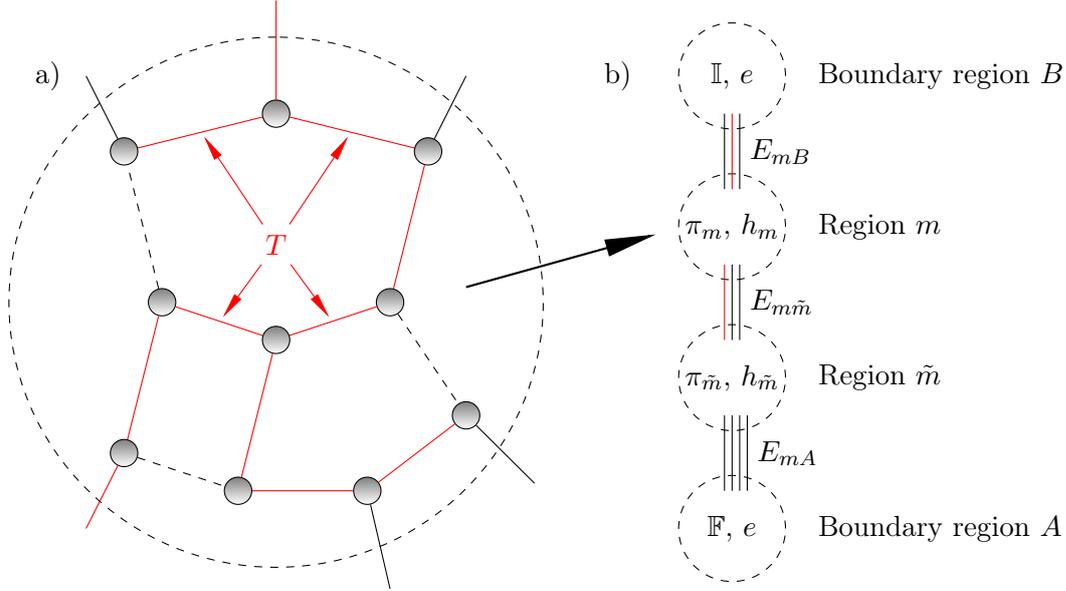
\begin{figure}
\begin{tikzpicture}
\begin{scope} [rotate=90]
\draw[dashed] (0,0) ellipse (100pt and 100pt);
\node (A1) at (-0.5,0) [circle,shade,draw] {};
\node (A2) at (0,1.5) [circle,shade,draw] {};
\node (A3) at (-2,2) [circle,shade,draw] {};
\node (A4) at (-2.5,0.5) [circle,shade,draw] {};
\node (A5) at (-1.5,-2.5) [circle,shade,draw] {};
\node (A6) at (0,-1.5) [circle,shade,draw] {};
\node (A7) at (2,-2) [circle,shade,draw] {};
\node (A8) at (2,2) [circle,shade,draw] {};
\node (A9) at (2.5,0) [circle,shade,draw] {};
\node (A10) at (-2.5,-1.2) [circle,shade,draw] {};
\draw[-,red] (A5)--(A10)--(A4)--(A1) -- (A2)-- (A3);
\draw[dashed] (A3) -- (A4);
\draw[-,red] (A1)--(A6) -- (A7) -- (A9) -- (A8); 
\draw[dashed] (A2) -- (A8);
\draw[dashed]  (A6) -- (A5);
\draw[red] (A3) -- (-3,2.5); 
\draw[-] (A5) -- (-2.4,-3.4); 
\draw[-] (A8) -- (3,2.5);  
\draw[-] (A7) -- (3,-2.5);
\draw[red] (A9) -- (4,0);
\draw[-] (A10) -- (-3.8,-1.5);

\node (MST) at (0.5,0)[above,red] {$T$};
\draw[->,red] (MST) -- (-0.2,0.7);
\draw[->,red] (MST) -- (-0.2,-0.7);
\draw[->,red] (MST) -- (2.2,-0.95);
\draw[->,red] (MST) -- (2.2,0.95);
\end{scope}
\draw[->, line width=0.8pt] (2.5,0.2) -- (5,0.9);

\begin{scope}[shift={(6,1)}]
\draw[dashed] (0,0) ellipse (20pt and 20pt) node[]{$\pi_m$, $h_m$};
\node (A) at (1,0)[right] {Region $m$};
\draw[dashed] (0,-4) ellipse (20pt and 20pt) node[]{$\mathbb{F}$, $e$};
\node (A) at (1,-4)[right] {Boundary region $A$};
\draw[dashed] (0,-2) ellipse (20pt and 20pt) node[]{$\pi_{\tilde{m}}$, $h_{\tilde{m}}$};
\node (A) at (1,-2)[right] {Region $\tilde{m}$};
\draw[dashed] (0,2) ellipse (20pt and 20pt) node[]{$\mathbb{I}$, $e$};
\node (A) at (1,2)[right] {Boundary region $B$};

\draw[red] (0,0.5)--(0,1.5);
\draw[-] (0.1,0.5)--(0.1,1.5)node [midway,right]{$E_{mB}$};
\draw[-] (-0.1,0.5)--(-0.1,1.5);
\draw[-] (0,-1.5)--(0,-0.5);
\draw[-] (0.1,-1.5)--(0.1,-0.5)node [midway,right]{$E_{m\tilde{m}}$};
\draw[red] (-0.1,-1.5)--(-0.1,-0.5);
\draw[-] (0,-3.5)--(0,-2.5);
\draw[-] (0.1,-3.5)--(0.1,-2.5);
\draw[-] (0.2,-3.5)--(0.2,-2.5)node [midway,right]{$E_{mA}$};
\draw[-] (-0.1,-3.5)--(-0.1,-2.5);
\end{scope}
\node (X) at (-3,3) {a)};
\node (Y) at (4.5,3) {b)};
\end{tikzpicture}
\caption{Coarse-graining of a permutation pattern by grouping neighbored nodes with same permutation symbol to one region a). A minimal graph $T$ connects all nodes to the boundary regions and reduces to $T^{\{V_m\}}$ in the coarse-grained graph $\Gamma^{\{V_m\}}$, sketched in b).}
\label{FIGregion}
\end{figure}

\begin{proof}
Let us find a minimal subgraph $T$ of the network graph $\Gamma$, which includes all nodes $V$ and to each node a path to a boundary node. Thus $T$ is a forest with leaves in the boundary and one can iteratively perform the integrals associated with its links in the amplitude (\ref{amplitude}). Starting from the links in the boundary $A$ and $B$ we find to each link $t\in T$ a different incident node $v_t$, which carries gauge parameters $h_{t_i}$. By integration with respect to the gauge parameter of this node, the delta functions associated with each face $f$ categorized by a link in $T$ vanishes. We denote the stranded subdiagram of $\mathcal{G}$, which is spanned by faces to the subgraph $T$, by $\mathcal{T}$ and by $\mathcal{G}/\mathcal{T}$ the diagram after omission of the faces to $t\in T$ as well as the evaluated gauge parameters $h_{t_i}$. It then holds:
\begin{align}
\int [\prod\limits_{l\in \mathcal{G}} dh_l] \prod\limits_{f \in \mathcal{G}} \delta(\prod_{l\in \partial f}h_l) &=\int [\prod\limits_{l\in \mathcal{G}/\mathcal{T}} dh_l][\prod\limits_{t\in T}\prod\limits_{i} dh_{t_i}] \prod\limits_{f \in \mathcal{G}/T} \delta(\prod_{l\in \partial f}h_l) \prod\limits_{f \in T} \delta(\prod_{l\in \partial f}h_l)
\\
&=\int [\prod\limits_{l\in \mathcal{G}/T} dh_l] \prod\limits_{f \in \mathcal{G}/T} \delta(\prod_{l\in \partial f}h_l)
\label{reducedamplitude}
\end{align}
After the integration procedure along the minimal subgraph $T$ we are left with a subdiagram $\mathcal{G}/\mathcal{T}$, which has a reduced face number $\widetilde{\Omega}[\mathcal{G}/\mathcal{T}]$. We use this face number as an upper bound of the divergence degree $\Omega[\mathcal{G}]$: 
\begin{align}
\Omega[\mathcal{G}]\, \leq \, \widetilde{\Omega}[\mathcal{G}/\mathcal{T}]=\widetilde{\Omega}[\mathcal{G}]-\sum\limits_{t=(t_1,t_2)\in T} (N-d(\pi_{t_1},\pi_{t_2}))
\label{divergenceestimation}
\end{align}
We rewrite (\ref{divergenceestimation}) by grouping neighbored vertices with the same permutation symbol to regions $V_m$. This coarse-graining procedure results in a graph $\Gamma^{\{V_m\}}$ (figure \ref{FIGregion}), where regions $m$ denote vertices and $E_{m\tilde{m}}$ notes the number of links between two regions. The minimal subgraph $T$ coarse-grains to $T^{\{V_m\}}$ by omitting links between nodes of the same region. With $d(\pi_{\pi_{v_1}},\pi_{v_2})=0$ between vertices in the same regions, we get:
\begin{align}
\label{divergenceestimation2.0}
\Omega[\mathcal{G}] \, & \leq  \, \widetilde{\Omega}[\mathcal{G}]-\sum\limits_{t\in T} (N-d(\pi_{t_1},\pi_{t_2})) - N|V|+ \sum_{(m,\tilde{m})\in T^{\{V_m\}}} d(\pi_m,\pi_{\tilde{m}})\\
&= N[|E\cup \partial \Gamma|-|V|] - \sum\limits_{(m,\tilde{m})\in \Gamma^{\{V_m\}}} E_{m\tilde{m}}d(\pi_m,\pi_{\tilde{m}})  + \sum\limits_{(m,\tilde{m})\in T^{\{V_m\}}} d(\pi_m,\pi_{\tilde{m}})
\label{divergenceestimation2}
\end{align}
While the maximal values of the second term in (\ref{divergenceestimation2}) have been identified in theorem \ref{THfaces}, we will optimize it here in combination with the third term in order to get a maximal upper bound of $\widetilde{\Omega}$ for a fixed network graph $\Gamma$. 
Since the second sum is always bigger than the third, both terms taken together are smaller or equal to zero. In the boundary situation of $Z_0^{(N)}$, this maximal bound is saturated for the pattern $\{\pi_v\}=\{\mathbb{I}\}$ with a divergence degree $\Omega[\{\mathbb{I}\}]=N[|E\cup \partial \Gamma|-|V|]$. Further maximal divergent processes are possible only if the second and third term vanish together and the inequality (\ref{divergenceestimation2}) remains straight. But this would imply a network graph which is coarse-grainable to a tree. The discussion of such graphs is the subject of Section \ref{APPmultipledivergence}. 

The different situation of $Z_A^{(N)}$ results in a sharper bound for the face number $\widetilde{\Omega}$. Under the assumption of a unique minimal surface $\sigma_{\text{min}}$, the second term is with theorem \ref{THfaces} at most $|\sigma_{\text{min}}|(1-N)$, which is reached in a situation of vanishing $T^{\{V_m\}}$. If $T^{\{V_m\}}$ does not vanish, the decrease of the second term needs to be compensated by the third term, where we observe a direct correspondence to other maximal cases in the situation of $Z_0^{(N)}$ in Section \ref{APPmultipledivergence}. In both boundary situations, there is thus the same number of $c$ different maximal divergent pattern.
\end{proof}

We note that unlike the discussed maxima of the face number $\widetilde{\Omega}$, the maximum of the divergence degree $\Omega$ is not unique and nontrivial examples are given in Section \ref{APPmultipledivergence}. However, we were able to show that the number $c$ of maxima is the same for both boundary conditions if one assumes a unique minimal surface $\sigma_{\text{min}}$. The multitude of the maxima has thus no influence on the entanglement entropy (\ref{centralapproximation}), which depends just on the quotient of the variables. 

\subsection{Ryu-Takayanagi Formula in the free theory}

In the limit of high dimensions $D(\Lambda)=\delta(e)$ of the leg space $\mathcal{H}$, the most divergent contributions determine the behavior of the fraction between the expectations of $Z_A^{(N)}$ and $Z_0^{(N)}$, which have been determined in the previous Sections. In the free theory the asymptotic behavior of the R\`enyi entanglement entropy (\ref{renyi}) is given by \cite{13Chi}:
\begin{align}
\label{RyuTakayanagiFree2}
\mathbb{E}[S_N(A)]\approx \frac{1}{N-1} \ln\left[\frac{\mathbb{E}[Z_A^{(N)}]}{\mathbb{E}[Z_0^{(N)}]}\right] & \approx \frac{1}{N-1}\ln\left[\frac{c\cdot \delta(e)^{N[|E\cup \partial \Gamma|-|V|]-|\sigma_{\text{min}}|(N-1)}}{c\cdot \delta(e)^{N[|E\cup \partial \Gamma|-|V|]}}\right]\\ \nonumber
& = |\sigma_{\text{min}}|\ln[\delta(e)]+\mathcal{O}(\frac{1}{\delta(e)})
\end{align}
Since no dependence on the parameter $N$ indicating the order of the R\`enyi-entropy is given in the limit $\delta(e)\gg1$, we directly apply the replica trick to recover the von-Neumann entanglement entropy $S(A)=\lim_{N\rightarrow 1}S_N(A)$. Equation (\ref{RyuTakayanagiFree2}) is thus the entanglement entropy within the approximation by local averaging in the free theory.

\

The proportionality of the entropy to the cardinality of the minimal domain wall $\sigma_{\text{min}}$ has a clear geometric interpretation, in the sense of discrete geometry, in the context of group field theory. The graph $\Gamma$ is the dual of a 2d simplicial complex. Each node is dual to a triangle and each link is dual to an edge of this complex, and the group field theory model endows the simplicial complex with dynamical geometric data. The length of each edge $\ell_j$
, in any given eigenstate of the length operator, is a function\footnote{The exact form of the function depends on the quantization map chosen to define the quantum theory.} of the irreducible representation $j_e$ associated to it, and to the dual link. If the quantum state is still defined on a fixed graph, but it is not an eigenstate of the length operator, then one has to average over the possible assignments of irreps $j_e$, with weights depending on the chosen state (i.e. on its decomposition into length eigenstates). We have
\begin{equation}
\text{Length}(\sigma_{\text{min}})= \sum_{e \in \sigma_{\text{min}}} \ell_e(j_e)=\braket{\ell_{j_e}} |\sigma_{\text{min}}|
\end{equation}
which can be interpreted as the length of a dual discrete minimal one-dimensional path. Therefore we can write $|\sigma_{\text{min}}|= \text{Length}(\sigma_{\text{min}})/\braket{\ell_{j_e}}$ and in this sense our result constitutes a Ryu-Takayanagi formula proposed in \cite{Ryu2006}, if we consider the path integral averaging over the open network $\Gamma$ as a simplified model of a bulk/boundary (spinfoam/network state) duality \cite{13Chi}. There are two further points to notice. First, our chosen quantum state fixes the parallel transports associated with the dual links to the identity and it is therefore maximally spread in the conjugate observables, which are in fact (including) the edge lengths associated to the same links. It is therefore a highly non-classical state to which it is not appropriate to associate a semi-classical geometric interpretation. A more appropriate choice, to this end, would be a coherent state peaking on both phase space variables on each link \cite{OritiCoherent,Bahr2007}
or, even better, a coherent state peaking on the collective variable $\text{Length}(\sigma_{\text{min}})$ as a whole \cite{OritiCollective}.
Second, a generic quantum state of the theory would also involve a superposition of combinatorial structures, i.e. a superposition of states associated to different graphs. In particular, this would be needed if the bulk is to admit a continuum geometric interpretation, going beyond the geometry associated to a (fixed) simplicial complex, which amounts to a drastic truncation of the allowed degrees of freedom of the fundamental (quantum gravity) model. In this case, one would have to understand the quantity $|\sigma_{\text{min}}|$ itself as the result of an average over such superposed graphs. Improving our derivation in both these directions would clearly be an interesting, and potentially important, development.

\section{Entropy Corrections from Group Field Interactions}
\label{SECint}

We are now interested in the possible modifications of \eqref{RyuTakayanagiFree2} induced by group field interaction terms. These interaction processes correspond to further stranded diagrams which contribute to the expectation value of $Z_{A/0}^{(N)}$. 

\subsection{Interaction processes}

In the free GFT calculations discussed above, the $d$-valence of the node tensors was arbitrary. For the interacting case, we fix the valence to $d=3$ and we specify the interaction kernel to be

\begin{align}
\mathcal{V}^{\text{sym}}(\{g_i^{(1)}\}\{g_i^{(2)}\}\{\overline{g}_i^{(1)}\}\{\overline{g}_i^{(2)}\}) = \int & \prod_{l=1}^4 dh_l \, \delta(h_1g_1^{(1)},h_3 \overline{g}_1^{(1)})\delta(h_1 g_2^{(1)},h_4\overline{g}_2^{(2)})\delta(h_1g_3^{(1)},h_2g_3^{(2)}) \nonumber \\
  & \delta(h_2g_1^{(2)},h_4\overline{g}_1^{(2)})\delta(h_2g_2^{(2)},h_3\overline{g}_2^{(1)})\delta(h_3\overline{g}_3^{(1)},h_4\overline{g}_3^{(2)})
\label{DEFv}
\end{align}

Stranded Feynman graphs contributing to $2N$-point functions are all possible combinations of the building blocks sketched in Figure \ref{FIGsymkernels}. Allowing for just one interaction vertex results already in a variety of possible combinations. If two fields of the interaction vertex propagate into each other, that is combining them with a propagator, we have an \emph{effective} propagator. Since the divergence degree of the resulting diagram is the same as that of a propagator,  the process could be captured in a mass renormalization. If two pairs of interacting fields propagate into each other, the process would disconnect and can be considered as a vacuum amplitude, which also does not contribute to the divergence degree of a diagram.
 
In this sense, we expect the only processes capable of extending the divergence degree to consist in the interaction of two incoming and two outgoing copies of network nodes. In the local averaging sector of the network statistics discussed in this Section, the $2N$-point functions of interest are the local averages $\mathbb{E}[(\rho^{(v)})^{\otimes N}]$, thus all interacting fields are restricted to copies of the same node $v$.

\subsection{Maximal divergent diagrams in the linear perturbation order}
\label{SECmaxdivint}

In the following, we shall estimate the divergence degree $\Omega$ of diagrams with one interaction process by determining the face number $\widetilde{\Omega}$, with the same strategy already used in Section \ref{SECfreelocal} for the free case. 

\begin{theorem}
\label{THEfacesint}
A pattern with a single interaction happening between two incoming and two outgoing fields of the same network node $v$ leads at most to the same face number $\widetilde{\Omega}$ as in a maximal case of the free theory. 
For the face number $\widetilde{\Omega}$ to reach the maximum, there must be a pattern of the free theory maximizing $\widetilde{\Omega}$, such that $v$ is incident to a link in the domain wall $\sigma$.
\end{theorem} 

\begin{proof}
Let $v\in V$ be an arbitrary node in a network and its neighbors given by $u_1,u_2,u_3$ connected with a link affecting the argument of $v$ with the respective index $(1,2,3)$. The by (\ref{DEFv}) chosen interaction block $\mathcal{V}$ connects only strands of the same argument index $i$, which again enables us to decompose the number of strands into separate link contributions, dependent on the local processes $\{\pi_i\}$ at the neighboring nodes. 
\begin{figure}
\begin{tikzpicture}[scale=0.8]
\foreach \i in {1,2,3}
{
\draw[-] (4,0.5+\i*0.2) -- (5,0.5+\i*0.2);
\draw[-] (11,0.5+\i*0.2) -- (12,0.5+\i*0.2);
\draw[-] (4,2+\i*0.2) -- (5,2+\i*0.2);
\draw[-] (11,2+\i*0.2) -- (12,2+\i*0.2);
\draw[-] (4,-2+\i*0.2) -- (5,-2+\i*0.2);
\draw[-] (11,-2+\i*0.2) -- (12,-2+\i*0.2);
}
\draw[-] (5,2.7) -- (5,2.1);
\draw[-] (5,1.2) -- (5,0.6);
\draw[-] (8,2.7) -- (8,2.1);
\draw[-] (8,1.2) -- (8,0.6);
\draw[-] (8,-1.3) -- (8,-1.9);
\draw[-] (11,2.7) -- (11,2.1);
\draw[-] (11,1.2) -- (11,0.6);
\draw[-] (11,-1.3) -- (11,-1.9);
\draw[-] (5,-1.3) -- (5,-1.9);
\draw[-] (8,-1.6) --(11,2.4);
\draw[-] (8,2.4)node[right]{$k$}--(11,-1.6)node[left]{$\bar{\pi}(k)$};
\node (A) at (4,3) {in};
\node (A) at (12,3) {out};
\node (A) at (9.5,3) {$\bar{\pi}$ at node $v$};
\node (A) at (3.5,2.4) {$\phi_{(1)}$};
\node (A) at (3.5,0.9) {$\phi_{(2)}$};
\node (A) at (3.5,-1.6) {$\phi_{(D)}$};
\node (A) at (12.5,2.4) {$\overline{\phi}_{(1)}$};
\node (A) at (12.5,0.9) {$\overline{\phi}_{(2)}$};
\node (A) at (12.5,-1.6) {$\overline{\phi}_{(D)}$};
\draw[] (5,0.9) -- (8,2.4) node[midway,above] {$2$} ;
\draw[] (5,2.4) -- (8,0.9);
\draw[]  (5,0.7) -- (8,0.7) node[midway,above] {$1$};
\draw[]  (5,2.6) -- (8,2.6) node[midway,above] {$1$} ;
\draw[] (5,1.1) to[bend right=120] (5,2.2) ;
\node (A) at (5,1.6) {$3$} ;
\node (A) at (8,1.6) {$3$} ;
\draw[] (8,1.1) to[bend left=120] (8,2.2); 
\draw[-] (5.2,1) ellipse (2pt and 13pt);
\draw[-] (7.8,1) ellipse (2pt and 13pt);
\draw[-] (5.2,2.3) ellipse (2pt and 13pt);
\draw[-] (7.8,2.3) ellipse (2pt and 13pt);
\draw[-] (6.5,-1.6) ellipse (2pt and 13pt);
\node (sledz) at (5,-0.3)[]{$\vdots$};
\node (sledz) at (8,-0.3)[]{$\vdots$};
\draw[]  (5,-1.4) -- (8,-1.4) node[midway,above] {} ;
\draw[]  (5,-1.6) -- (8,-1.6) node[midway,above] {} ;
\draw[]  (5,-1.8) -- (8,-1.8) node[midway,above] {} ;
\begin{scope}[shift={(18,0.5)}]
\node (A) at (0,0) [circle,shade,draw] {$v$};
\node (sledz) at (0,-0.8)[]{$\bar{\pi}$};
\draw[-] (A) -- (-2,0) node[midway,above]{1} node[left]{$\pi_1$} ;
\draw[-] (A) -- (2,1) node[midway,above]{2} node [right]{$\pi_2$};
\draw[-] (A) -- (2,-1) node[midway,above]{3} node[right]{$\pi_3$} ;
\end{scope}
\end{tikzpicture}
\caption{Stranded graph of a local pattern with an interaction happening at the network node $v$. After relabeling of the network copies, the first two incoming copies of $v$ participate in the interaction and a further permutation $\bar{\pi}$ ensures the generality of the process.}
\label{FIGpermutationintsym}
\end{figure}
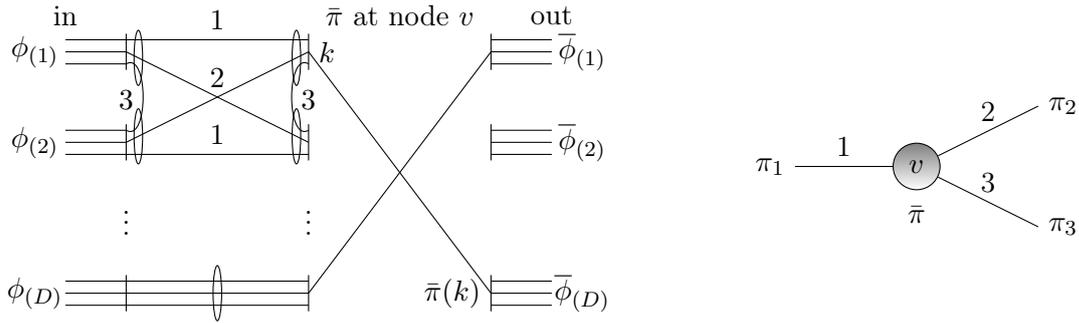
Let us assume the fields participating in the interaction process are of the first two copies, where a permutation operator to the symbol $\bar{\pi}$ acts on the outgoing fields (figure \ref{FIGpermutationintsym}). It is sufficient to discuss just this case, since we can renumber the copy index of the network, thus each $\bar{\pi}$ represents $\binom{N}{2}$ equivalent interaction processes.

We assume now for a given $\bar{\pi}$ a higher face number than for all reference free propagations at $v$ with symbol $\pi$, including the choice $\pi=\bar{\pi}$. The additional faces in the interacting case must result from link contributions with field indices 2 and 3 since the link contribution with index 1 is in both cases given by $\chi(\bar{\pi}\circ \pi_1^{-1})$.

For the link 3 to contribute not less in the interacting case than in the free case, $\bar{\pi}\circ \pi_3^{-1}$ must contain a cycle including positions $1$ and $2$, since only in that case the interaction of the index $3$ does not decrease the number of cycles in $\bar{\pi}\circ \pi_3^{-1}$ compared to the free case. Link $2$ contributes in the interacting case by one more than in the free case if and only if the additional action of $\mathbb{F}_2\otimes\mathbb{I}_{N-2}$ increases the number of cycles, that is, if and only if in $\bar{\pi}\circ \pi_2^{-1}$ the positions 1 and 2 are included in a cycle.

We consider now the choice $\pi=(\mathbb{F}_2\otimes \mathbb{I}_{N-2})\circ \bar{\pi}$ for the reference free propagation. Since $\bar{\pi}\circ \pi_2^{-1}$ and $\bar{\pi}\circ \pi_3^{-1}$ contain cycles connecting the first two positions, in $\pi \circ \pi_2^{-1}$ and $\pi \circ \pi_2^{-1}$ the first two positions are in different cycles, thus it holds:
\begin{align}
\chi(\pi \circ \pi_2^{-1})  =\chi(\bar{\pi}\circ \pi_2^{-1})+1 \quad &, \quad \chi(\pi \circ \pi_3^{-1}) =\chi(\bar{\pi}\circ \pi_3^{-1})+1 \quad ,\\
\chi(\pi \circ \pi_1^{-1})& \geq \chi(\bar{\pi}\circ \pi_1^{-1})- 1 \\
  \Rightarrow\quad   \sum_{i=1}^3 \chi(\pi \circ \pi_i^{-1})& \geq \sum_{i=1}^3 \chi(\bar{\pi} \circ \pi_i^{-1})+1
\end{align}
The choice $\pi=(\mathbb{F}_2\otimes\mathbb{I}_{N-2})\circ \bar{\pi}$ for a free reference process results thus at least in the same face number compared to $\bar{\pi}$ in the interacting theory. There is thus no $\bar{\pi}$, such that the associated interaction process leads to a higher face number $\widetilde{\Omega}$ compared to all other free propagations. 

Furthermore, there is no interaction pattern maximizing $\widetilde{\Omega}$ in case of $\pi_1=\pi_2=\pi_3$, since at most $N-1$ closed strands correspond to the index $3$ in an interaction pattern, whereas a free propagation with $\pi=\pi_1$ would result in $N$ closed strands corresponding to each index. Hence, if at $v$ an interaction process leads to a maximal face number $\widetilde{\Omega}$, there must be a different propagation process at two neighbors of $v$, thus $v$ would be incident to a domain wall $\sigma$ in a reference pattern in the free theory.
\end{proof}

Since for the boundary situation of $Z_0^{(N)}$ no domain walls appear in free permutation pattern maximizing $\widetilde{\Omega}$, local interaction can only lead to maximal face numbers in the situation $Z_A^{(N)}$. From the proof of theorem \ref{THEfacesint} we identify $2\binom{N}{2}$ possible interaction processes happening at one of the $2|\sigma_{\text{min}}|$ nodes incident to the minimal surface $\sigma_{\text{min}}$, which maximize the face number $\widetilde{\Omega}$. For the case of symmetric network states under study, these patterns will not result in an increasing divergence degree compared to the free patterns, as we will show in the following. 
\begin{theorem}
\label{THEdivergenceint}
A pattern with a single interaction happening between two incoming and two outgoing fields of the same network node $v$ has at most the same divergence degree $\Omega$ as in a reference case of the free theory, where the interaction at $v$ is replaced by a free propagation.
\end{theorem}

\begin{proof}
Let us again model the interaction process at $v$ by $\bar{\pi}$, as sketched in figure \ref{FIGpermutationintsym}.
From theorem \ref{THEfacesint} we already know, that the face number $\widetilde{\Omega}$ of the diagram with an interaction does not exceed the degree of a reference free propagation happening at $v$. In the free propagating case there are two gauge parameters associated with the first two node copies, thus at most two evaluations of delta functions can be induced by them. 

Let us now assume, that the interaction process leads to a higher divergence degree than for all replacements by free propagations, which implies that there can be at most one evaluation of delta functions induced by the interaction block. Since all delta functions associated with one network link can always be evaluated, the symbol $\bar{\pi}^{-1}\circ \pi_1$ would have a cycle containing the first two indices and the symbol $\bar{\pi}^{-1}\circ \pi_2$ would not. But then a replacement of $\bar{\pi}$ by $(\mathbb{F}_2\otimes \mathbb{I}_{N-2})\circ \bar{\pi}$ would increase the number of faces by two. If our assumption of maximal divergence was correct, there have to be at least three evaluations in the modified interaction process, such that the modified process does not exceed the divergence degree. But the modified interaction process cannot be maximal divergent with three evaluations since we find a reference free propagation with at least the same face number and at most two evaluations.
\end{proof}

Theorem \ref{THEdivergenceint} enables us to follow the previous arguments in the optimization of the face number $\widetilde{\Omega}$, since the divergence degree of patterns in the free theory cannot be increased by local modifications with interaction processes. We thus just consider the pattern in the free theory, which maximizes the divergence degree $\Omega$, and study the impact of the modification by a local interaction process. Since only in this case the face number can stay constant, the modification needs to take place at a node incident to a domain wall.

\begin{theorem}
\label{THEfinal}
Let $\Gamma=(V,E\cup \partial\Gamma)$ be a connected network graph, which is not coarse-grainable to a tree and which has a unique minimal set $\sigma_{\text{min}} \subset E$ separating the boundary regions $A$ and $B$. Then, each diagram with a single local interaction process has a smaller divergence degree compared to the unique maxima in the free theory.
\end{theorem}

\begin{proof}
In theorem \ref{THdivergence} we determined the unique pattern of the free theory maximizing the divergence degree under the same assumptions as here. The dominant pattern to $Z_0^{(N)}$ does not have a domain wall, thus with theorems \ref{THEfacesint} and \ref{THEdivergenceint} all patterns with a single local interaction process have subleading divergence degree. Although the dominant pattern contributing to $Z_A^{(N)}$ contains a domain wall $\sigma_{\text{min}}$, an interaction happening at a node incident to $\sigma_{\text{min}}$ always results in three evaluations induced by the interaction block and the diagram would be subleading.
\end{proof}

Theorem \ref{THEfinal} therefore proves for a broad class of network graphs, that the local permutation pattern of the free theory are maxima of the divergence degree, also if we include pattern with  -- up to one -- happening interaction process. For the simple case of a single $\mathcal{O}(\lambda)$ GFT interaction term the free theory result 
\begin{align}\nonumber
\mathbb{E}[S_N(A)]\approx |\sigma_{\text{min}}|\ln[\delta(e)]+\mathcal{O}(\frac{1}{\delta(e)})
\end{align}
is not modified. As we will discuss in the next Section, there are, however, special network architectures, which are coarse-grainable to a tree and allow for multiple maxima of $\Omega$, even if the minimal surface $\sigma_{\text{min}}$ is unique. 

One main point to notice about the above result is the general link between the computed entropy and the divergences of the quantum amplitudes of the group field theory model we have been using. This is important from a conceptual standpoint, as we will discuss later on. From a more technical perspective, it is also important because the perturbative divergences of group field theory models are a well-explored subject \cite{Rivasseau2011, Carrozza:2016vsq}, and in particular the divergences of topological GFTs, as the Boulatov mode we have used, are well-understood 
\cite{FreidelRegularization,Bonzom,BonzomTwisted,Carrozza2014,
Geloun}
They are associated to 3-cells of the 3d cellular complex identified by each GFT Feynman diagram, or to the vertices of the dual 3d simplicial complex, and they have been fully classified. Such generic divergences do not appear in our calculation, however, because Feynman diagrams at order $\lambda$ only involve a single interaction vertex and no bubble (in fact, there is no real 3d dynamics at such linear order, even in the sense of topological gravity, thus any physical interpretation of our present result should be attempted with caution), but it is clear that they will become crucial when going beyond this crude approximation, and that the technical tools to do so are already available. 

Concerning the same perturbative approximation of the underlying GFT model that we have relied on (and that the whole spin foam literature, for example, relies on too), another important cautionary remark should be made. Assuming it is tantamount to assuming that the quantity we are trying to evaluate is analytic in the GFT coupling constant, and thus well approximated already at low orders. From the point of view of spacetime structures, this means assuming that their role is well approximated by simple cellular complex (made of few vertices, edges, faces etc), {\it before any coarse-graining of the same complexes, and of the associated quantum amplitudes}. This may not be the case. In fact, it is not expected to be the case in much GFT (and tensor models) literature, where instead the goal is to extract emergent continuum gravitational (thus geometric) physics from the collective behavior of the underlying microscopic degrees of freedom \cite{OritiBronstein,OritiCondensate,OritiLQG,Rivasseau2016,
Rivasseau2012}.
If the relation between the Renyi entropy of our states and the divergences of the underlying GFT is generic, as we expect, the former is probably not analytic in the coupling constant, and its value will be dictated by the most divergent contributions to the Feynman expansion of the GFT model, which are obviously growing with the number of interaction vertices involved, thus it will be ultimately dominated by higher powers of $\lambda$.
This is one more reason to go beyond the approximation adopted in the present work.

\subsection{Networks coarse-grainable to a tree}
\label{APPmultipledivergence}
We recall the upper bound (\ref{divergenceestimation2}) of the divergence degree in case of free local propagations, where neighboring nodes with same propagation pattern were coarse-grained to regions $m$. It has been shown in Section \ref{SECfreelocal} and \ref{SECmaxdivint}, that up to the linear order of the interaction a pattern is maximal divergent if and only if for this upper bound is maximal and straight. A maximum of the upper bound is attained for the pattern also maximizing the face number $\widetilde{\Omega}$, which consists of a single region with $\mathbb{I}$ in the boudary case of $Z_0^{(N)}$ and an additional region with $\mathbb{F}$ in case of $Z_A^{(N)}$. Derivations from this pattern result in the same upper bound (\ref{divergenceestimation2}), only if the following stays constant:
\begin{align}
-\sum_{m\tilde{m}}E_{m\tilde{m}}d(\pi_m,\pi_{\tilde{m}})+\sum_{m\tilde{m}\in T^c}d(\pi_m,\pi_{\tilde{m}})
\label{twoterms}
\end{align}
In case of multiple minimal surfaces $\sigma_{\text{min}}$ this bound stays constant for the pattern discussed in Appendix \ref{APPmultipleminimal}, where $T^c$, which is the minimal subgraph along which all associated gauge freedoms are evaluated, increases.

If the minimal surface $\sigma_{\text{min}}$ is unique, (\ref{twoterms}) can just attain further maxima in both boundary situations, if the decrease of the first term is compensated by the increase of the second. However, since $E_{m\tilde{m}}\geq 1 $ for $(m,\tilde{m}) \in T^c$, this directly implies a tree structure of the coarse-grained graph. As sketched in figure \ref{FIGexamplecoarsetree}, we thus find further maximal upper bounds (\ref{twoterms}) in case of a network graph, which is coarse-grainable to a tree. The number of maxima of the bound (\ref{twoterms}) is furthermore equal in the boundary situation of $Z_0^{(N)}$ and $Z_A^{(N)}$. Since in each case all delta functions associated with the links between different regions will drop out in the amplitude (\ref{amplitude}) by parameter evaluation, the choice of the permutation symbol at each region does not modify $\Omega$. The less amount of faces in this case is compensated by a remaining freedom of the parameter $h_m$ in each region of the coarse-grained tree structure.

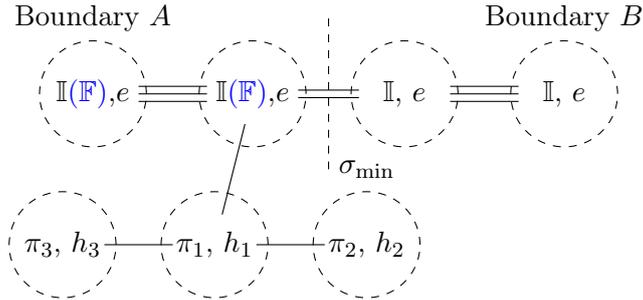
\begin{figure}[t]
\centering
\begin{tikzpicture}
\node (X) at (-4.1,1) {Boundary $A$};
\node (X) at (2.1,1) {Boundary $B$};

\draw[dashed] (0,0) ellipse (20pt and 20pt) node[]{$\mathbb{I}$, $e$};
\draw[dashed] (-4.1,0) ellipse (20pt and 20pt) node[]{$\mathbb{I}$\textcolor{blue}{($\mathbb{F}$)},$e$};
\draw[dashed] (-2,0) ellipse (20pt and 20pt) node[]{$\mathbb{I}$\textcolor{blue}{($\mathbb{F}$)},$e$};
\draw[dashed] (2.1,0) ellipse (20pt and 20pt) node[]{$\mathbb{I}$, $e$};
\draw[dashed] (-2.5,-2) ellipse (20pt and 20pt) node[]{$\pi_{1}$, $h_{1}$};
\draw[dashed] (-0.5,-2) ellipse (20pt and 20pt) node[]{$\pi_{2}$, $h_{2}$};
\draw[dashed] (-4.5,-2) ellipse (20pt and 20pt) node[]{$\pi_{3}$, $h_{3}$};
\draw[-] (-1.05,-2)--(-1.95,-2);
\draw[-] (-3.05,-2)--(-3.95,-2);

\draw[-] (-3.5,0)--(-2.6,0);
\draw[-] (-3.5,0.1)--(-2.6,0.1);
\draw[-] (-3.5,-0.1)--(-2.6,-0.1);

\draw[-] (-1.4,-0.05)--(-0.6,-0.05);
\draw[-] (-1.4,0.05)--(-0.6,0.05);

\draw[-] (0.6,0)--(1.5,0);
\draw[-] (0.6,0.1)--(1.5,0.1);
\draw[-] (0.6,-0.1)--(1.5,-0.1);

\draw[dashed] (-1,1)--(-1,-1) node[right]{$\sigma_{\text{min}}$};
\draw[-] (-2.1,-0.4)--(-2.4,-1.6);
\end{tikzpicture}

\caption{Network with a tree structure after a coarse-graining procedure. The special topology enables maximal divergent patterns with arbitrary symbol $\pi_m$.}
\label{FIGexamplecoarsetree}
\end{figure}

We can exploit the independence of the divergence degree $\Omega$ on the links connecting the coarse-grained tree to construct maximal divergent pattern with a local interaction. Let us therefore assume the region with symbol $\pi_1$ consisting just of one network node $u$. Since in this case all links incident to $u$ do not contribute to the divergence degree, any local process happening at $u$, thus also interaction processes, result in a maximal divergent diagram. 

\section{Discussion}
\label{SECdiscussion}

\subsection{Entanglement and geometry}

Symmetric group field networks carrying discrete geometries in their spin-network decomposition were taken as a state class of central interest. Along with the pre-geometric interpretation of the fundamental quanta of a group field theory \cite{OritiQMB}, we understand such network states as a collection of abstract quanta of space connected by entanglement patterns partially reflecting symmetry and topology of a quantum discrete geometries \cite{Dona, Bianchi, Mele}. In this picture, the quantum-many-body approach to quantum geometry has allowed to import new quantitative tools to investigate the behavior of the quantum geometry states in non-perturbative quantum gravity \cite{Han}. As an explicit example in this sense, the Ryu-Takayanagi formula, generally intended as a proportionality between entanglement entropies of boundary states to the surface of minimal areas in a dual bulk state \cite{TakayanagiBuch},  can be used as a guiding principle to select states with an interesting semiclassical limit \cite{Hamma}. 

The kinematic states associates to quantum geometries in the background independent approach to Quantum Gravity \cite{01Ash} consist of spin-network states with the structure of a tensor network. Tensorial entanglement and building of discrete geometries by network states have a natural correspondence, made precise in \cite{13Chi}. It is however not clear from first principles if the geometric interpretation of holographic network models in terms of the Ryu-Takayanagi formula is reflected by such discrete geometries. A connection between spin-networks and holography is established in \cite{Han}, where network states representing the holographic duality are derived from a coarse-graining of spin-network states. Spin networks are thereby representing bulk degrees of freedom, where the coarse-grained tensor network represents the boundary state. Also within this setup, the Ryu-Takayanagi proposal holds in the semiclassical regime provided by a limit of increasing scale in the coarse-graining procedure \cite{Han}. 

A crucial ingredient for establishing the holographic properties of such network states in quantum gravity states may be found in the random character of the tensor networks, as suggested in \cite{Hayden}. The group-field theoretic derivation for the typical entanglement entropy for a generalized open spin-network, as given in \cite{13Chi} points in this direction. The underlying assumptions of a free theory and independent averaging at each network node capture just a small corner of possible amplitudes, nevertheless we find further evidence for the dominance of this corner in this work.


We were able to prove the validity of the independence assumption within a class of cases in Section \ref{SEClocalevidence}. 
%
However, the central new aspect faced in this work consists in the computation of the correction terms to the entanglement entropy possibly induced by deviations from the Gaussian random tensor description, introduced by considering an interacting group field theory, hence a polynomially perturbed Gaussian measure for the tensor fields. Building on theorem \ref{THEdivergenceint} and two assumptions on the network graph, which is not coarse-grainable to a tree and possesses a unique minimal surface, the main result of the work consists in the proof, in theorem \ref{THEfinal}, that the linear order correction of the perturbation series produces no leading amplitudes in the average of $Z_0^{(N)}$ and $Z_A^{(N)}$, thus the Ryu-Takayanagi formula is not modified.



The arguments used are strongly based on the dominance of the stranded Feynman diagram by the contraction structure (figure \ref{FIGfullvariable}), which is fixed by the network topology and the boundary conditions of the variables. Already at the level of the free theory, it would be interesting to see what changes if the kinetic term is modified from a simple delta function to something more general. For higher order perturbations, on the other hands, the diagrams will be dominated by the bulk structure, given by interaction vertices, and we expect significant changes in the leading order amplitudes. To calculate the impact of these sectors, one needs to control the divergent degrees of such diagrams, similarly to what has been done in previous renormalization studies \cite{Bonzom}. 

One interesting general feature of our results is the relation we found between the Renyi entropy and the (perturbative) divergences of the GFT model within which this is computed. This implies a direct link between continuum physics and geometry, to the extent to which it is captured by the Renyi entropy and the Ryu-Takayanagi formula,  and the renormalization group flow of the fundamental quantum gravity model, thus the collective, many-body physics of its basic entities. Beyond the technical points we have already discussed, this relation is evocative of (and clearly consistent with) the general perspective that sees continuum spacetime and geometry as emergent from the collective behavior of the  fundamental quantum gravity degrees of freedom encoded in these models. This perspective, indeed, motivates a large part of the literature and in particular the one concerned with GFT renormalization (both perturbative and non-perturbative). Such overall coherence between the specific GFT realization of the emergent spacetime (and geometry) perspective and the ideas inspiring the \lq geometry from entanglement\rq  scenario, within which the Ryu-Takayanagi result plays such an important role, is very remarkable, albeit still tentative. We take it as a further motivation to proceed along the research direction followed in the present work.

\subsection{Modification of the holographic entanglement scaling}

Finally, an interesting remark concerns the role of symmetry in this result. Along with the result obtained in \cite{13Chi}, throughout the work we have treated the network states and their entanglement variables on the base of a probabilistic distribution on group fields satisfying a \emph{symmetry constraint}. Although the symmetry constraint is important for a reformulation of group field networks in terms of spin-network states \cite{13Chi}, dropping it offers an interesting perspective on their statistics. By considering analogous probability measures on more general group fields, that do not satisfy the closure constraint discussed in Section \ref{SECsymmetricgroupfields} and therefore defining an action by a propagation and interaction kernel by dropping the gauge integrations in (\ref{choiceK}) and (\ref{DEFv}) one finds:
\begin{align}
\mathcal{K}^{\text{non-sym}}(g_1,...,g_d,\overline{g}_1,...,\overline{g}_d)= & \prod_{i=1}^d \delta( g_i\circ \overline{g}_i^{-1})
\label{choiceKunsym}
\\
\mathcal{V}^{\text{non-sym}}(\{g_i^{(1)}\}\{g_i^{(2)}\}\{\overline{g}_i^{(1)}\}\{\overline{g}_i^{(2)}\}) = &  \delta(g_1^{(1)}, \overline{g}_1^{(1)})\delta( g_2^{(1)},\overline{g}_2^{(2)})\delta(g_3^{(1)},g_3^{(2)}) \nonumber 
\\
  & \delta(g_1^{(2)},\overline{g}_1^{(2)})\delta(g_2^{(2)},\overline{g}_2^{(1)})\delta(\overline{g}_3^{(1)},\overline{g}_3^{(2)})
\label{choiceVunsym}
\end{align}
Both kernels thus remain the combinatorial structure with the only difference lying in the missing group averaging with respect to gauge parameters $h$, which would enforce symmetry properties. The expectation values $\mathbb{E}_{\text{non-sym}}[Z_{A/0}^{(N)}]$, now averaging on general group fields, are analogously expanded into stranded graphs as in figure \ref{FIGfullvariable}. Without integration of the gauge parameters, which amounts to already trivial face holonomies, the amplitudes of stranded graphs $\mathcal{G}$ reduces to:
\begin{align}
\mathcal{A}[\mathcal{G}]=\prod\limits_{f \in G} \delta(e)=\delta(e)^{\widetilde{\Omega}[\mathcal{G}]}
\label{amplitudeunsym}
\end{align}
The divergent degree of any diagram equals the face number $\widetilde{\Omega}[\mathcal{G}]$ and we can directly apply theorems 
\ref{THEconnected}, \ref{THEtreenetwork}, \ref{THfaces} and \ref{THEfacesint} to determine the asymptotic behavior of $\mathbb{E}_{\text{non-sym}}[Z_{A/0}^{(N)}]$. Assuming a network graph $\Gamma=(V,E\cup \partial\Gamma)$ with a disjoint partition $A\cup B=\partial\Gamma$ such that the minimal surface $\sigma_{\text{min}}$ is unique and the graph cannot be coarse-grained to a tree, in the sector of local averaging one finds 
\begin{align}
\mathbb{E}_{\text{non-sym}}[Z_{A}^{(N)}]=& \left[1+2\lambda\binom{N}{2}|\sigma_{\text{min}}|+\mathcal{O}(\lambda)^2)\right]\left[1+\mathcal{O}\frac{
1}{\delta(e)})\right]\delta(e)^{|V|-(N-1)|\sigma_{\text{min}}|}
\\
\mathbb{E}_{\text{non-sym}}[Z_{0}^{(N)}]=& \delta(e)^{|V|}
\end{align}
The expected $N$th R\`enyi entanglement entropy (\ref{centralapproximation}) is thus estimated as:
\begin{align}
\mathbb{E}_{\text{non-sym}}[S_N(\rho_A)]\approx & \ln[\delta(e)]|\sigma_{\text{min}}|-\frac{\ln\left[1+2\lambda \binom{N}{2}\right]|\sigma_{\text{min}}|]}{N-1}\\
\approx & |\sigma_{\text{min}}| \left[\ln[\delta(e)]-\lambda N \right]
\label{scalingcorrection}
\end{align}
In the non-symmetric case the linear order correction modifies the asymptotical scaling of the R\`enyi entanglement entropy with the area of a minimal surface and establishes therefore a new Ryu-Takayanagi proportionality \cite{Ryu2006}. The von-Neumann entropy $S(A)$ \cite{Nielsen} corresponds to the limit $N\rightarrow 1$ the R\`enyi entropy, where the proportionality gets in facts corrected by the coupling constant $\lambda$ of the perturbed group field theory.  The active role of the GFT dynamics in the rescaling of the area proportionality factor is intriguing. A more refined analysis of the modified holographic entanglement scaling consists an interesting quesiton for further investigation.

\section*{Acknowledgements}
The authors are grateful to the support of the Albert Einstein Institute, where the main part of this work was realized. A. Goe{\ss}mann also acknowledges the recent support from the MATH+ research center and the Fritz Haber Institute. M. Zhang acknowledges the support from the A. von Humboldt Foundation.

\appendix

\section{Networks with multiple maxima of the face number in the free theory}
\label{APPmultipleminimal}

The unique pattern maximizing the face number $\widetilde{\Omega}$ in the boundary conditions of $Z_0^{(N)}$ is with theorem \ref{THEconnected} the association of $\mathbb{I}$ to each node. For the different boundary conditions of $Z_A^{(N)}$ theorem \ref{THfaces} identifies the unique maximum of $\widetilde{\Omega}$ with additional assumption of a unique minimal surfaces $\sigma_{\text{min}}$ separating the regions $A$ and $B$. Dropping this assumption can give rise to different maxima, if the inequalities (\ref{ineq1},\ref{ineq2}) hold straight. All links between nodes with different permutation symbols have therefore to be included in $|\sigma_{\text{min}}|$ disjoint paths $P_k$ between boundary $A$ and $B$ and the triangle equation needs to hold straight for the permutations along each path. This directly implies the separation of the network into regions of constant permutations by different minimal surfaces $\sigma_{\text{min}}$.

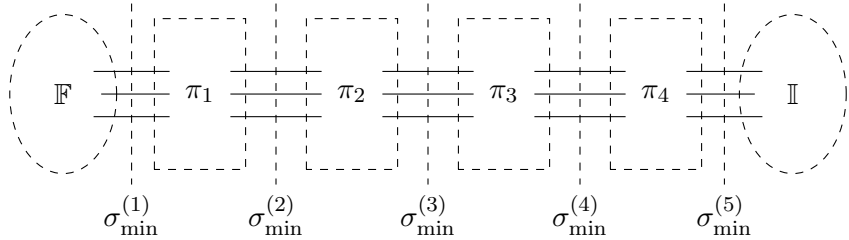
\begin{figure}[t]
\centering
\begin{tikzpicture}
\draw[dashed] (0,0) rectangle (1.2,2);
\draw[dashed] (2,0) rectangle (3.2,2);
\draw[dashed] (4,0) rectangle (5.2,2);
\draw[dashed] (6,0) rectangle (7.2,2);
\draw[-] (0.2,1)--(-0.7,1);
\draw[-] (0.2,1.3)--(-0.8,1.3);
\draw[-] (0.2,0.7)--(-0.8,0.7);
\draw[-] (1,1)--(2.2,1);
\draw[-] (1,1.3)--(2.2,1.3);
\draw[-] (1,0.7)--(2.2,0.7);
\node (X) at (0.6,1) {$\pi_1$};
\draw[-] (3,1)--(4.2,1);
\draw[-] (3,1.3)--(4.2,1.3);
\draw[-] (3,0.7)--(4.2,0.7);
\node (X) at (2.6,1) {$\pi_2$};
\draw[-] (5,1)--(6.2,1);
\draw[-] (5,1.3)--(6.2,1.3);
\draw[-] (5,0.7)--(6.2,0.7);
\node (X) at (4.6,1) {$\pi_3$};
\draw[-] (7,1)--(7.9,1);
\draw[-] (7,1.3)--(8,1.3);
\draw[-] (7,0.7)--(8,0.7);
\node (X) at (6.6,1) {$\pi_4$};
\draw[dashed] (-1.2,1) ellipse (20pt and 30pt) node[]{$\mathbb{F}$};
\draw[dashed] (8.4,1) ellipse (20pt and 30pt) node[]{$\mathbb{I}$};
\draw[dashed] (-0.3,2.2)--(-0.3,-0.2) node[below] {$\sigma_{\text{min}}^{(1)}$};
\draw[dashed] (1.6,2.2)--(1.6,-0.2) node[below] {$\sigma_{\text{min}}^{(2)}$};
\draw[dashed] (3.6,2.2)--(3.6,-0.2) node[below] {$\sigma_{\text{min}}^{(3)}$};
\draw[dashed] (5.6,2.2)--(5.6,-0.2) node[below] {$\sigma_{\text{min}}^{(4)}$};
\draw[dashed] (7.5,2.2)--(7.5,-0.2) node[below] {$\sigma_{\text{min}}^{(5)}$};
\end{tikzpicture}
\caption{Coarse-grained network with multiple minimal surfaces $\sigma_{\text{min}}$. If the regions of a permutation pattern are separated by minimal surfaces and the triangle equation holds straigt for the associated permutation processes, the face number is maximial.}
\label{FIGcascades}
\end{figure}

An example of a sequence of permutation symbols $(\pi_i)_{i=0}^{N-1}\subset \mathcal{S}_N$, for which the inequality (\ref{ineq2}) holds straight, is the combination of the cyclic element $\mathbb{F}_{N-i}\in \mathcal{S}_{N-i}$ permuting the first $N-i$ copies with the trivial element $\mathbb{I}\in \mathcal{S}_{i}$ for the last $i$ copies:
\begin{align}
\pi_i:=\mathbb{F}_{N-i}\otimes \mathbb{I}_{i}\in \mathcal{S}_{N}
\end{align}
Since $d(\pi_i,\pi_{j})=j-i$ holds for this sequence, we have for a collection of indices $0= i_1\leq i_2\leq ...\leq i_l= N-1$:
\begin{align}
\sum_{k=1}^l d(\pi_{i_k},\pi_{i_{k+1}})=(N-1)=d(\mathbb{F},\mathbb{I})
\end{align}
With this the triangle equations (\ref{ineq2}) would hold straight for a path $P_k$ along nodes with permutation symbols $\pi_{i_k}$. A pattern with maximal face number $\widetilde{\Omega}$, sketched in figure \ref{FIGcascades}, is thus given by from $A$ to $B$ numerated regions separated by different choices of $\sigma_{\text{min}}$, where elements of the sequence $\{\pi_{i_k}\}$ determine the local propagations. The number of such patterns depends besides the number of different minimal surfaces also on their arrangement, since this determines the connectivity of the resulting regions and thus influence the triangle equations to be staten along the disjoint paths. However, a multitude of $c$ different maximal divergent diagrams contributing to the expectation of $Z_A^{(N)}$ results in an offset of $\frac{ln[c]}{N-1}$ in the entanglement entropy (\ref{RyuTakayanagiFree2}). Since the offset is independent from the leg space dimension $D$, a multitude of different maximal pattern does not influence the asymptotic entanglement entropy in the limit of high $D$.

\raggedright
\bibliographystyle{plain}
\bibliography{./Literatur1803}

\end{document}

%% file: gmydefs.tex
\newcommand{\be}{\begin{equation}}
\newcommand{\ee}{\end{equation}}
\newcommand{\bea}{\begin{eqnarray}}
\newcommand{\eea}{\end{eqnarray}}

\newcommand{\bit}{\begin{itemize}}
\newcommand{\eit}{\end{itemize}}

\renewcommand{\d}{\delta}

